\documentclass{wp}

\usepackage{booktabs}
\usepackage{tabularx}
\newcolumntype{C}{>{\centering\arraybackslash}X}

\newcommand{\lott}{\prob_\rms}
\newcommand{\usd}{\mathrel{\unrhd}}

\newcommand{\meu}{\textsc{meu}}

\renewcommand{\wpr}{\succsim}

\title{Randomization and ambiguity perception}
\author{Yutaro Akita}
\address{Department of Economics, Pennsylvania State University}
\email{ytrakita@gmail.com}
\author{Kensei Nakamura}
\address{Graduate School of Economics, Hitotsubashi University}
\email{kensei.nakamura.econ@gmail.com}
\date{November 2, 2026}
\jelcodes{D81}
\thanks{We are grateful to Kalyan Chatterjee and Yuhta Ishii
for their invaluable advice.
We also thank
\begin{itemize*}[nolabel]
  \item Nageeb Ali
  \item Mira Frick
  \item Ryota Iijima
  \item Shaowei Ke
  \item Philipp Sadowski
  \item Norio Takeoka
  \item Shohei Yanagita
\end{itemize*}
for their helpful comments.}

\begin{document}

\begin{abstract}
  Ambiguity-averse decision makers typically dislike
  not only the presence of ambiguous events but also their increase,
  contrary to what standard ambiguity models predict.
  We axiomatically study such a decision maker.
  She avoids ex ante randomization over prospects
  since it increases the number of relevant ambiguous events
  without providing a hedge against uncertainty.
  Our axioms lead to a representation in which the decision maker behaves
  as if she optimizes her ambiguity perception at a cost.
  We show the uniqueness of the representation and
  conduct comparatives of attitudes toward ambiguity and its increase.
  This identification is not achieved
  without considering ex ante randomization.
\end{abstract}

\maketitle

\section{Introduction}\label{sec:intro}

Thought experiments introduced by \citet{Ellsberg1961} suggest that
ambiguity influences decisions.
The decision maker (\dm) usually dislikes
betting on an ambiguous event, one that has an unknown probability.
Several models have been developed
to explain aversion to ambiguous prospects compared to unambiguous ones.
We examine the behavioral implications of ``ambiguity aversion''
from a different perspective,
based on comparisons between prospects
involving more ambiguous events that affect the {\dm}'s final payoff
and those involving fewer such events.

The importance of this aspect is demonstrated by the ``reflection example''
of \citet{Machina2009}:
A box contains $100$ balls.
Of these, $50$ are either red or blue, and
the other $50$ are either green or purple.
No information is provided on the exact number of balls of each color.
A ball is drawn at random from the box.
Consider bets $f_1$, \dots, $f_4$ on the color of the drawn ball,
as described in \cref{tab:ref}.
A natural preference pattern for an ``ambiguity-averse'' {\dm} is
to prefer $f_2$ over $f_1$ and $f_3$ over $f_4$,
which is also confirmed experimentally \citep{LHaridonPlacido2010}.
The intuition is as follows.
The ambiguity relevant to the payoff from $f_2$ (resp.~$f_3$)
concerns only the relative likelihood of green and purple (resp.~red and blue).
In contrast, ambiguity about all the colors is relevant to $f_1$ and $f_4$.
Thus, $f_2$ and $f_3$ have fewer relevant ambiguous events than $f_1$ and $f_4$,
which leads the {\dm} to avoid $f_1$ and $f_4$.
Nevertheless, most of the ambiguity-sensitive decision models preclude
this intuitive pattern \citep{Machina2009,BLP2011}.

\begin{table}[t]
  \centering
  \renewcommand{\arraystretch}{1.3}
  \begin{tabularx}{80mm}{CCCCC}\toprule
          & \multicolumn{2}{c}{$50$ balls} & \multicolumn{2}{c}{$50$ balls} \\
          \cmidrule(lr){2-3}\cmidrule(lr){4-5}
          & Red   & Blue  & Green  & Purple \\ \hline
    $f_1$ & $100$ & $200$ & $100$  & $0$    \\
    $f_2$ & $100$ & $100$ & $200$  & $0$    \\
    $f_3$ & $0$   & $200$ & $100$  & $100$  \\
    $f_4$ & $0$   & $100$ & $200$  & $100$  \\ \bottomrule
  \end{tabularx}
  \caption{Reflection example.}
  \label[table]{tab:ref}
\end{table}

This paper presents a decision model
that incorporates key insights from the reflection example:
the {\dm} dislikes not only the presence of ambiguity
but also increases in the number of relevant ambiguous events.
The literature has tied aversion to ambiguity
with a preference for (objective) randomization
\emph{after} the state realization,
which we refer to as \emph{ex post randomization}.
A popular manifestation is
\citets{Schmeidler1989} \textit{uncertainty aversion} axiom.
The idea is that ex post randomization smoothes
payoff variability across states that the ambiguity-averse {\dm} dislikes.
In contrast, attitudes toward \emph{ex ante randomization}
(i.e., one \emph{before} the state realization)
play a key role in our analysis.
We observe that ex ante randomization increases the relevant ambiguous events,
whereas it does not smooth payoffs.
This feature leads to a new prediction, \axmref{axm:eaar};
that is, the {\dm} dislikes ex ante randomization.

To illustrate, let $A$ be the event
that the temperature in Athens is greater than or equal to 20$^\circ$C, and
let $B$ be the corresponding event for the temperature in Beijing.
The {\dm} is unfamiliar with these places, and so $A$ and $B$ are ambiguous.
For each $E \in \{A, B\}$,
let $f_E$ be the bet on $E$ whose stakes are \$$100$ and \$$0$, and
let \$$c_E$ be the \dm's certainty equivalent of $f_E$.
Also, consider the certainty equivalent \$$c_P$
of the half--half lottery $P$ over $f_A$ and $f_B$.
Ambiguity in $f_A$ (resp.~$f_B$) that is relevant to the \dm's payoff is
only on the temperature in Athens (resp.~Beijing).
By contrast, in $P$, the temperatures in both Athens and Beijing are relevant.
Moreover, regardless of the realization of $P$,
the {\dm} will always face an ambiguous bet, either $f_A$ or $f_B$,
so $P$ does not provide a hedge against ambiguity.
Thus, $P$ only complicates the dependence of outcomes on ambiguous events,
which leads to $c_P \le \max \{c_A, c_B\}$.
Namely, ex ante randomization is not beneficial.

We formalize this intuition
by analyzing the \dm's preference on the domain of \citet{AnscombeAumann1963}:
the set of lotteries over acts---%
an act maps a state to a distribution of outcomes.
We provide an axiomatic foundation and an identification result
for a new decision model,
the \emph{costly ambiguity perception} model (\cref{def:cap}),
in which the {\dm} behaves as if optimizing her ambiguity perception at a cost.
The {\dm} has a von Neumann--Morgenstern (vNM) function $u$
over distributions of outcomes.
Ambiguity perceptions are represented by sets of priors and
are associated with costs.
Given an ambiguity perception $M$, the {\dm} evaluates each act $f$
according to the maxmin expected utility (\meu) model
of \citet{GilboaSchmeidler1989};
that is, the value of $f$ is $\min_{\mu \in M} \int u \cmpf f \diff \mu$.
The {\dm} chooses her ambiguity perception $M$ from a feasible set $\bM$
to maximize the associated expected {\meu} value minus the cost $c(M)$.
Her preference over lotteries of acts is represented by the function
\begin{equation}
  P \mapsto \max_{M \in \bM} \Bigl[
    \int \Bigl(\min_{\mu \in M} \int u \cmpf f \diff \mu\Bigr) \diff P(f)
    - c(M)\Bigr].
\end{equation}
Thus, the \dm's ambiguity perception is endogenously chosen for each lottery,
whereas it is fixed in the {\meu} model.
The maximization reflects the \dm's aversion
to an increase in relevant ambiguity:
her tuning of ambiguity perception is easier with fewer ambiguous events.
\cref{ex:ref} in \cref{subsec:rep} confirms this point
by showing that the costly ambiguity perception model explains
the typical pattern in the reflection example under a natural assumption.
We interpret the maximization as the \dm's mental process,
such as contemplation, belief distortion, or ambiguity-attitude adjustment.

The first main result (\cref{thm:rep}) is an axiomatic characterization.
Our key axioms control the attitude toward ex ante randomization.
\axmref*{axm:eaar} formalizes the observation in the temperature example above:
the {\dm} is averse to ex ante randomization
since it increases the number of relevant ambiguous events.
\axmref*{axm:ica} postulates that
her ranking is independent of (ex ante) mixing with constant acts
when the mixing weights are kept unchanged.
To elaborate, using the same notation as before,
for each $E \in \{A, B\}$,
let $Q^x_E$ be the lottery that yields $f_E$ with probability $\lambda$ and
\$$x$ with probability $1 - \lambda$.
The only relevant ambiguity in $Q^x_E$ is on event $E$,
which is independent of $x$.
Thus, as long as the weight $\lambda$ on $f_E$ is the same,
the particular value of $x$ should not affect the \dm's comparison
between $Q^x_A$ and $Q^x_B$.
This idea motivates \axmref{axm:ica},
which requires that if $Q^x_A$ is preferred to $Q^x_B$,
then $Q^y_A$ is preferred to $Q^y_B$ for any other monetary payoff \$$y$.
In addition to the two key axioms and other standard requirements,
we impose two axioms (\axmref{axm:aepr,axm:irtc})
that deal with the attitude toward the timing of randomization.

We explore in \cref{subsec:id} the extent to which
the parameters of the costly ambiguity perception model are identifiable.
We focus on the uniqueness of the costs and the feasible sets
of ambiguity perceptions.
Our second main result (\cref{thm:id}) establishes that
the costs are essentially unique in the class of convex cost structures.
We also identify the smallest feasible set
of ambiguity perceptions that represent the same preference.
Moreover, we can construct the convex costs
from certainty equivalent data, which are observable in principle.

Our uniqueness result exploits the richness of the domain
stemming from ex ante randomization;
\cref{ex:id} illustrates that it does not hold without ex ante randomization.
We construct a costly ambiguity perception preference
that is behaviorally indistinguishable from an expected utility preference
unless choices over lotteries of acts are observed, but
that differs in attitudes toward ex ante randomization and randomization timing.
Such sharp identifiability of parameters
is another advantage of considering the richer domain,
in addition to capturing aversion to increases in ambiguous events.

Building on the identification result,
we describe three comparative statics results in \cref{sec:cmp}.
Although the costs and choices of ambiguity perceptions are unobservable,
we can compare them across {\dm}s through their preferences over lotteries.
We show that
\begin{enumerate*}
  \item
    a more ambiguity-averse {\dm} has higher costs;
  \item
    a {\dm} more averse to additional ambiguity chooses
    less refined ambiguity perceptions;
  \item
    a {\dm} more averse to ex ante randomization tends
    to choose different perceptions for different lotteries.
\end{enumerate*}
The second comparative is conceptually novel and
directly related to our motivation
for studying attitudes toward increases in ambiguity.

Finally, \cref{sec:disc} discusses special cases,
the properties of parameters, and related models.
\cref{subsec:meu} compares our model
with a version of \citets{GilboaSchmeidler1989} {\meu} model
at the axiomatic level.
\cref{subsec:spec} pins down the additional axioms
to characterize three special cases
that have been studied in a framework
without ex ante randomization \citep{PTX2026,Sinander2025,CFIL2022}.
\cref{subsec:int} studies
the properties of intersections of feasible ambiguity perceptions.
\cref{subsec:dual} considers the ``dual'' version of our model,
in which the {\dm} chooses her ambiguity perception
to minimize (rather than maximize) the {\meu} value at a cost.
This version generalizes \citets{KeZhang2020} double {\meu} model
by adding a cost term.
We characterize it by replacing \axmref{axm:eaar}
with \axmref{axm:eapr} in \cref{thm:rep}.
Among these related models,
the one proposed by \citet{PTX2026} is conceptually the closest,
as it also considers the costly choice of ambiguity perception,
based on the Choquet expected utility model \citep{Schmeidler1989}.
We further compare their model with ours and
discuss other related models in \cref{subsec:lit}.

All proofs are in the \hyperref[sec:prelim]{Appendix}.

\section{Model}\label{sec:model}

\subsection{Primitives}\label{subsec:prim}

Given any set $Y$,
denote by $\lott(Y)$
the set of finitely supported probability measures on $Y$;
for each $y \in Y$, denote by $\Dirac[y]$ the Dirac measure at $y$;
denote by $\Dirac[Y]$ the set of Dirac measures on $Y$.

Let $\Omega$ be a finite set of \emph{states} with $\abs{\Omega} \ge 2$ and
let $X$ be a set of \emph{consequences}.
An \emph{act} is a function from $\Omega$ to $\lott(X)$.
Let $\cF$ be the set of acts.
With an abuse of notation,
identify each $p \in \lott(X)$ with the constant act of value $p$.
The \dm's preference is modeled as a binary relation $\wpr$ on $\lott(\cF)$.
Each member of $\lott(\cF)$ is called a \emph{lottery}.
Denote by $\ipr$ and $\spr$ the symmetric and asymmetric parts of $\wpr$,
respectively.

For each $(\lambda, (f, g)) \in [0, 1] \times \cF^2$,
define $\lambda f + (1 - \lambda)g \in \cF$ as
\begin{equation}
  [\lambda f + (1 - \lambda)g](\omega)(Z)
  = \lambda f(\omega)(Z) + (1 - \lambda)g(\omega)(Z)
  \quad \forall (\omega, Z) \in \Omega \times 2^X.
\end{equation}
That is, in $\lambda f + (1 - \lambda)g$,
a mixture happens \emph{after} a state is realized.
Refer to this type of mixture as \emph{ex post randomization}.

For each $(\lambda, (P, Q)) \in [0, 1] \times \lott(\cF)^2$,
define $\lambda P + (1 - \lambda)Q \in \lott(\cF)$ as
\begin{equation}
  [\lambda P + (1 - \lambda)Q](F) = \lambda P(F) + (1 - \lambda)Q(F)
  \quad \forall F \in 2^\cF.
\end{equation}
That is, in $\lambda P + (1 - \lambda)Q$,
a mixture happens \emph{before} a state is realized.
Refer to this type of mixture as \emph{ex ante randomization}.
In particular, $\lambda \Dirac[f] + (1 - \lambda)\Dirac[g]$
is the lottery that yields $f$ and $g$
with probability $\lambda$ and $1 - \lambda$, respectively
(see the left tree in \cref{fig:tree}).
It is distinguished from the degenerate lottery
$\Dirac[\lambda f + (1 - \lambda)g]$
at the ex post randomization of $f$ and $g$
(see the right tree in \cref{fig:tree}).

\begin{figure}[b]
  \centering
  \begin{tikzpicture}
    \tikzset{
      > = latex,
      x = 0.5pt,
      y = 0.5pt,
      level distance = 60,
      font = \small,
      node/.style = {
        draw, solid, rectangle, text centered,
        minimum height = 22, minimum width = 22,
      },
      act/.style = { node, fill = cyan!20 },
      lott/.style = { node, rounded corners, fill = yellow!30 },
      level 1/.style = { sibling distance = 100, dashed, -> },
      level 2/.style = { sibling distance = 50, solid, - },
      level 3/.style = { sibling distance = 20, dashed, -> },
    }
    \node at (0, 0) [lott] {$\lambda \Dirac[f] + (1 - \lambda)\Dirac[g]$}
      child {
        node [act] {$f$}
        child {
          node [lott] {$f(\omega_1)$}
          child { node {} }
          child { node {} }
          edge from parent node [left] {$\omega_1$}
        }
        child {
          node [lott] {$f(\omega_2)$}
          child { node {} }
          child { node {} }
          edge from parent node [right] {$\omega_2$}
        }
        edge from parent node [left] {$\lambda$\enspace}
      }
      child {
        node [act] {$g$}
        child {
          node [lott] {$g(\omega_1)$}
          child { node {} }
          child { node {} }
          edge from parent node [left] {$\omega_1$}
        }
        child {
          node [lott] {$g(\omega_2)$}
          child { node {} }
          child { node {} }
          edge from parent node [right] {$\omega_2$}
        }
        edge from parent node [right] {\enspace$1 - \lambda$}
      }
    ;
    \tikzset{ level 2/.append style = { sibling distance = 130 } }
    \node at (480, 0) [lott] {$\Dirac[\lambda f + (1 - \lambda)g]$}
      child {
        node [act] {$\lambda f + (1 - \lambda)g$}
        child {
          node [lott] {$\lambda f(\omega_1) + (1 - \lambda)g(\omega_1)$}
          child { node {} }
          child { node {} }
          child { node {} }
          edge from parent node [left] {$\omega_1$}
        }
        child {
          node [lott] {$\lambda f(\omega_2) + (1 - \lambda)g(\omega_2)$}
          child { node {} }
          child { node {} }
          child { node {} }
          edge from parent node [right] {$\omega_2$}
        }
        edge from parent node [left] {$1$}
      }
    ;
  \end{tikzpicture}
  \caption{Ex ante and ex post randomization.}
  \label{fig:tree}
\end{figure}

\subsection{Representation}\label{subsec:rep}

Let $\prob(\Omega)$ be the set of probability measures on $\Omega$,
endowed with the total variation distance.
Let $\bK$ be the set of nonempty compact convex subsets of $\prob(\Omega)$,
endowed with the Hausdorff metric.
Each member of $\bK$ represents an ambiguity perception:
the set of priors the {\dm} believes possible.
A \emph{vNM function}
is a nonconstant mixture linear real-valued function on $\lott(X)$.
With a vNM function $u$ and an ambiguity perception $M$,
the {\meu} model evaluates each act $f$
as $\min_{\mu \in M} \int u \cmpf f \diff \mu$.

Our representation captures the {\dm}
who optimally chooses her ambiguity perception at a cost.
The costs are represented
by a function over feasible ambiguity perceptions.
An extended real-valued function is \emph{grounded} if its infimum is zero.
A \emph{cost structure} is a pair $(\bM, c)$
of a nonempty compact subset of $\bK$ and
a lower semicontinuous grounded extended real-valued function on $\bM$ such that
$M \subseteq M'$ implies $c(M) \ge c(M')$ for each $(M, M') \in \bM^2$.

\begin{definition}\label{def:cap}
  A \emph{costly ambiguity perception representation of $\wpr$}
  is a pair $\tuple{u, (\bM, c)}$ of
  a surjective vNM function and a cost structure such that
  the real-valued function $U$ on $\lott(\cF)$ of the form
  \begin{equation}
    U(P)
    = \max_{M \in \bM} \Bigl[
      \int \Bigl(\min_{\mu \in M} \int u \cmpf f \diff \mu\Bigr) \diff P(f)
      - c(M)\Bigr]
  \end{equation}
  represents $\wpr$.
\end{definition}

A \emph{costly ambiguity perception preference} is a binary relation
on $\lott(\cF)$ that has a costly ambiguity perception representation.

In a costly ambiguity perception representation $\tuple{u, (\bM, c)}$,
the set $\bM$ represents the \dm's feasible set of ambiguity perceptions, and
$c$ specifies the cost of refining her ambiguity perception to each $M \in \bM$.
The {\dm} chooses her ambiguity perception
to maximize the expected {\meu} value of a lottery minus the cost.
By the definition of a cost structure,
a smaller (more refined) ambiguity perception is more expensive.
While the {\meu} model evaluates each alternative
based on a single ambiguity perception fixed as a parameter,
the costly ambiguity perception model endogenously chooses the \dm's perception
from a feasible set.

We offer several interpretations of the representation.
Under each of the interpretations,
the representation can be seen as a reduced form
of the \dm's process of subjective optimization,
which is not directly observable:
\begin{enumerate}
  \item
    \textit{Contemplation.}
    The {\dm} engages in contemplation
    to exclude unreasonable priors from her candidate ambiguity perception.
    She chooses the amount of contemplation
    to balance the benefits of a finer ambiguity perception against its cost.
    See also \citet{ErginSarver2010,ErginSarver2015},
    who study costly contemplation by a {\dm} who is uncertain about her tastes.
  \item
    \textit{Optimism.}
    The {\dm} optimistically distorts her ambiguous beliefs
    depending on the situation she faces.
    Her beliefs are chosen
    as a solution to the maximization appearing in the representation.
    The cost term captures the \dm's mental costs
    associated with each distortion.
    See, for example, \citet{BrunnermeierParker2005} and \citet{Kovach2020}
    for models of belief distortion by a Bayesian {\dm}.
  \item
    \textit{Ambiguity-attitude adjustment.}
    In the {\meu} model,
    each set of priors may be viewed as the \dm's ambiguity attitude:
    the bigger the set is, the more ambiguity averse she is.
    Depending on the context, the size of the set may represent
    the \dm's fear of ambiguity or her confidence about probabilistic laws.
    Facing an ambiguous prospect,
    the {\dm} attempts to choose the optimal degree of ambiguity aversion
    accordingly;
    however, this process entails cognitive costs.
    See also \citet{PTX2026} for related discussion and
    more specific parametrizations of the model.
\end{enumerate}

A costly ambiguity perception preference exhibits
aversion to increases in relevant ambiguous events
due to the perception-choice stage.
With fewer such events, it is easier for the {\dm}
to tailor her choice of ambiguity perception to the alternative.
Indeed, this model can accommodate
the natural pattern observed in the reflection example (\cref{tab:ref})
discussed in the \nameref{sec:intro}.

\begin{example}\label{ex:ref}
  Let $B = G = [-1/4, 1/4]$.
  Identify each $(\mu_\rmB, \mu_\rmG) \in [0, 1/2]^2$ with the prior such that
  the probability of drawing blue is $\mu_\rmB$ and
  drawing green is $\mu_\rmG$.
  For each $(\beta, \gamma) \in [0, 1]^2$,
  let $M(\beta, \gamma) = \set{(1/4 + \beta b, 1/4 + \gamma g)
  \mvert (b, g) \in B \times G}$.
  Each of the parameters $\beta$ and $\gamma$ controls the size of
  ambiguity perception $M(\beta, \gamma)$:
  the smaller $\beta$ (resp.~$\gamma$) is,
  the more refined her ambiguity perception about blue (resp.~green) is.
  Let $\bM = \set{M(\beta, \gamma) \mvert (\beta, \gamma) \in [0, 1]^2}$,
  let $c \colon [0, 1]^2 \to \SR_+$ be
  a symmetric strictly decreasing lower semicontinuous grounded function, and
  let $\theta \in (0, \infty)$.
  Let $U$ be the utility function over acts
  corresponding to the cost structure $(\bM, \theta c)$.
  Then,
  \begin{align}
    U(f_1)
    &= \max_{(\beta, \gamma) \in [0, 1]^2} \Bigl[
      100 \min_{(b, g) \in B \times G} (1 + \beta b + \gamma g)
      - \theta c(\beta, \gamma)\Bigr] \\
    &= 100 - \min_{(\beta, \gamma) \in [0, 1]^2}
      [25(\beta + \gamma) + \theta c(\beta, \gamma)], \\
    U(f_2)
    &= \max_{(\beta, \gamma) \in [0, 1]^2} \Bigl[
      100 \min_{(b, g) \in B \times G} (1 + 2\gamma g)
      - \theta c(\beta, \gamma)\Bigr]
    = 100 - \min_{\gamma \in [0, 1]} (50\gamma + \theta c(1, \gamma)), \\
    U(f_3)
    &= \max_{(\beta, \gamma) \in [0, 1]^2} \Bigl[
      100 \min_{(b, g) \in B \times G} (1 + 2\beta b)
      - \theta c(\beta, \gamma)\Bigr]
    = 100 - \min_{\beta \in [0, 1]} (50\beta + \theta c(\beta, 1)), \\
    U(f_4)
    &= \max_{(\beta, \gamma) \in [0, 1]^2} \Bigl[
      100 \min_{(b, g) \in B \times G} (1 + \beta b + \gamma g)
      - \theta c(\beta, \gamma)\Bigr] \\
    &= 100 - \min_{(\beta, \gamma) \in [0, 1]^2}
      [25(\beta + \gamma) + \theta c(\beta, \gamma)].
  \end{align}
  By the monotonicity of $c$, for each $(\beta, \gamma) \in [0, 1]^2$,
  \begin{equation}\label{eq:ref_ineq}
    25(\beta + \gamma) + \theta c(\beta, \gamma)
    \ge 50\min \{\beta, \gamma\} + \theta c(1, \min \{\beta, \gamma\}).
  \end{equation}
  Thus, $U(f_1) \le U(f_2)$.
  By symmetry, $U(f_3) \ge U(f_4)$.
  In \eqref{eq:ref_ineq}, equality holds if and only if $\beta = \gamma = 1$,
  in which case $50\min \{\beta, \gamma\} + \theta c(1, \min \{\beta, \gamma\})
  = 50$.
  Hence, since $\theta \mapsto \min_{\gamma \in [0, 1]}
  (50\gamma + \theta c(1, \gamma))$ is concave and
  $\lim_{\theta \downarrow 0} \min_{\gamma \in [0, 1]}
  (50\gamma + \theta c(1, \gamma)) = 0$,
  there exists $\bar \theta \in \SR_{++}$ such that
  $\theta < \bar \theta$ implies
  $\min_{\gamma \in [0, 1]} (50\gamma + \theta c(1, \gamma)) < 50$,
  so $U(f_1) < U(f_2)$ and $U(f_3) > U(f_4)$.

  The symmetry of the cost function captures the informational symmetry.
  If $\theta$ is not too high ($\theta < \bar \theta$),
  then when the {\dm} faces $f_2$ or $f_3$,
  she can set one of the parameters $\beta$ or $\gamma$ equal to $1$ and
  can concentrate on optimizing the other parameter;
  on the other hand, at $f_1$ and $f_4$,
  since the {\dm} has to control both $\beta$ and $\gamma$,
  she has to pay higher costs than at $f_2$ or $f_3$
  to achieve the same {\meu} value.
  Finally, the {\dm} exhibits absolute ambiguity aversion
  since $\bM$ satisfies the condition of \cref{cor:absaa} in \cref{subsec:int}.
  In particular, the preference is consistent with the Ellsberg paradox
  at the same time.
\end{example}

An example of a costly ambiguity perception preference that generates
the typical pattern in the reflection example is provided by \citet{PTX2026},
who examine a special case of our model.
We discuss in \cref{sec:disc} the relationship between our model and theirs.
\cref{ex:ref} generalizes the example they provide,
which makes it easier to understand the force leading to the typical pattern.

\section{Foundations}\label{sec:fnd}

\subsection{Characterization}\label{subsec:char}

We introduce several axioms on $\wpr$.

We begin with a standard requirement.
The relation $\wpr$ is \emph{nondegenerate}
if there exists $(P, Q) \in \lott(\cF)^2$ such that $P \spr Q$; and
\emph{mixture continuous} if for each $(P, Q, R) \in \lott(\cF)^3$,
the sets $\set{\lambda \in [0, 1]
\mvert \lambda Q + (1 - \lambda)R \wpr P}$ and
$\set{\lambda \in [0, 1] \mvert P \wpr \lambda Q + (1 - \lambda)R}$ are closed.

\begin{axiom}[Regularity]\label{axm:reg}
  The relation $\wpr$ is
  nondegenerate, complete, transitive, and mixture continuous.
\end{axiom}

The following axiom requires that the {\dm} respects statewise dominance.

\begin{axiom}[Monotonicity]\label{axm:mon}
  For each $(\lambda, P, (f, g)) \in [0, 1] \times \lott(\cF) \times \cF^2$,
  if $\Dirac[f(\omega)] \wpr \Dirac[g(\omega)]$ for each $\omega \in \Omega$,
  then $\lambda \Dirac[f] + (1 - \lambda)P
  \wpr \lambda \Dirac[g] + (1 - \lambda)P$.
\end{axiom}

We introduce two axioms on the attitude toward the timing of randomization,
which are also used by \citet{KeZhang2020}.
First, while ex post randomization can help the {\dm} hedge against uncertainty
by reducing payoff variability across states,
ex ante randomization does not:
the {\dm} will always end up with an ambiguous act after the randomization.
This difference in hedging power motivates
the {\dm} to prefer the former to the latter.

\begin{axiom}[Attraction to ex post randomization]\label{axm:aepr}
  For each $((\kappa, \lambda), P, (f, g))
  \in [0, 1]^2 \times \lott(\cF) \times \cF^2$,
  we have $\kappa \Dirac[\lambda f + (1 - \lambda)g] + (1 - \kappa)P
  \wpr \kappa[\lambda \Dirac[f] + (1 - \lambda)\Dirac[g]]
  + (1 - \kappa)P$.
\end{axiom}

Second, if an act is mixed with a constant act,
even ex post randomization of the two does not provide hedging:
it preserves the pattern of payoff variability in the act.
Thus, the {\dm} has no reason to prefer ex post randomization,
since \axmref{axm:aepr} is motivated by its hedging power.
The following axiom states that
the timing of such a mixture does not matter for the {\dm}.

\begin{axiom}[Indifference to randomization timing of constant acts]
  \label{axm:irtc}
  For each $((\kappa, \lambda), P, f, p)
  \in [0, 1]^2 \times \lott(\cF) \times \cF \times \lott(X)$,
  we have $\kappa \Dirac[\lambda f + (1 - \lambda)p] + (1 - \kappa)P
  \ipr \kappa[\lambda \Dirac[f] + (1 - \lambda)\Dirac[p]] + (1 - \kappa)P$.
\end{axiom}

Combined with \axmref{axm:aepr},
this axiom implies that
the {\dm} strictly prefers ex post randomization to ex ante randomization
only when the two acts to be mixed both involve ambiguity.

We then proceed to two axioms on the attitude toward ex ante randomization.
Besides being useless for hedging against uncertainty,
ex ante randomization complicates the dependence of outcomes
on ambiguous events by increasing the number of relevant events,
as illustrated in the \nameref{sec:intro}.
Consequently, ex ante randomization is unfavorable to the {\dm}.

\begin{axiom}[Ex ante aversion to randomization]\label{axm:eaar}
  For each $(\lambda, (P, Q)) \in [0, 1] \times \lott(\cF)^2$,
  if $P \wpr Q$, then $P \wpr \lambda P + (1 - \lambda)Q$.
\end{axiom}

The second axiom in this category postulates that
the ranking between two lotteries is independent
of their common unambiguous parts.

\begin{axiom}[Independence of constant acts]\label{axm:ica}
  For each $(\lambda, (P, Q), (p, q))
  \in [0, 1] \times \lott(\cF)^2 \times \lott(X)^2$,
  if $\lambda P + (1 - \lambda)\Dirac[p]
  \wpr \lambda Q + (1 - \lambda)\Dirac[p]$,
  then $\lambda P + (1 - \lambda)\Dirac[q]
  \wpr \lambda Q + (1 - \lambda)\Dirac[q]$.
\end{axiom}

In each lottery of the form $\lambda P + (1 - \lambda)\Dirac[p]$,
replacing the constant act $p$ with another constant act $q$ does not change
the relevant ambiguous events in the lottery.
Thus, the ranking between $\lambda P + (1 - \lambda)\Dirac[p]$ and
$\lambda Q + (1 - \lambda)\Dirac[p]$ should be preserved
even when $p$ is replaced with $q$.
That is, \axmref{axm:ica} requires that
changes in parts irrelevant to ambiguity do not affect the \dm's ranking.

The final axiom ensures that
the range of the vNM function is the entire real line.

\begin{axiom}[Unboundedness]\label{axm:ubd}
  For each $(p, q) \in \lott(X)^2$ with $\Dirac[p] \spr \Dirac[q]$,
  there exists $(r, s) \in \lott(X)^2$ such that
  $\half \Dirac[r] + \half \Dirac[q] \wpr \Dirac[p]$ and
  $\Dirac[q] \wpr \half \Dirac[p] + \half \Dirac[s]$.
\end{axiom}

The axioms introduced in this section are all the behavioral implications
of costly ambiguity perception preferences.

\begin{theorem}\label{thm:rep}
  The relation $\wpr$ satisfies
  \axmref{axm:reg,axm:mon,axm:aepr,axm:irtc,axm:eaar,axm:ica,axm:ubd}
  if and only if it has a costly ambiguity perception representation.
\end{theorem}

The sufficiency part of the proof of \cref{thm:rep} exploits
results from convex analysis and the notion of niveloids
developed in \citet{CMMR2014} (see \cref{subsec:cvx,subsec:nv}).
To do so, we transform each lottery to a continuous real-valued function
on a topological space.
We begin by finding a vNM function $u$ that represents
the restriction of $\wpr$ to $\Dirac[\lott(X)]$.
For each $P \in \lott(\cF)$,
denote by $P_u \in \lott(\SR^\Omega)$
the pushforward of $P$ under $f \mapsto u \cmpf f$.
Then, we define a topological space $\bU$ and
a mixture linear mapping $m \mapsto m^\vee$ from $\lott(\SR^\Omega)$ to
the set of bounded continuous real-valued functions on $\bU$
(see \cref{subsec:cost} for the precise definitions).
Our construction suggests that $P_u^\vee = Q_u^\vee$ implies $P \ipr Q$.
Using \axmref{axm:eaar,axm:ica},
we can construct a convex niveloid $W$
on $\set{m^\vee \mvert m \in \lott(\SR^\Omega)}$ such that
$P \mapsto W(P_u^\vee)$ represents $\wpr$.
Thus, by the mixture linearity of $m \mapsto m^\vee$,
there exist a pair $(\cV, \gamma)$ of
a set of real-valued functions on $\SR^\Omega$ and
a lower semicontinuous real-valued function on it such that
$W$ has the form
\begin{equation}
  W(m^\vee) = \max_{V \in \cV} \Bigl(\int V \diff m - \gamma(V)\Bigr).
\end{equation}

An important step is to restrict $\cV$ to the class of {\meu} functions.
This step relies on \cref{prop:nv_var_me} in \cref{subsec:nv},
which is our main technical innovation.
It states that we can take $\cV$ as the set satisfying
\begin{align}
  &\int V \diff l \ge \int V \diff \tilde l \quad \forall V \in \cV \\
  &\iff W(\lambda l^\vee + (1 - \lambda)m^\vee)
  \ge W(\lambda \tilde l^\vee + (1 - \lambda)m^\vee)
  \quad \forall (\lambda, m) \in (0, 1] \times \lott(\SR^\Omega).
\end{align}
Combined with \axmref{axm:aepr,axm:irtc},
we can show that each member of $\cV$ can be written as
the {\meu} function associated with a closed convex subset of $\prob(\Omega)$.
A result related to \cref{prop:nv_var_me} was shown by \citet{CMMR2015},
but it is not applicable in our setting.
In comparison, \cref{prop:nv_var_me} extends their result to
a convex niveloid on a more general domain%
---a set with possibly empty interior.

Although \cref{thm:rep} is related
to \citets{KeZhang2020} representation theorem
for a preference on a similar domain,
our proof strategy is different from theirs.
They characterize the model
where the {\dm} chooses the worst ambiguity perception without costs and
evaluates each lottery using the {\meu} function
(for a formal definition, see \cref{subsec:dual}).
Their proof relies on the relationship between $\wpr$ and
its subrelation $\wpr^*$, which we define in \cref{subsec:id}.
They show that $\wpr^*$ has a ``multi-{\meu} representation'' and
the original relation $\wpr$ can be recovered
through a ``cautious completion'' of $\wpr^*$.
This leads to a minimization of {\meu} functions
over candidate ambiguity perceptions without costs.
In our case, the argument of cautious completion is not applicable, and
we instead use properties of niveloids to obtain a cost structure.

\subsection{Identification}\label{subsec:id}

The restrictions imposed on cost structures in \cref{def:cap} are not enough
to identify them.
Indeed, if a cost structure $(\bM, c)$ satisfies
$\lambda M + (1 - \lambda)M' \in \bM$ and
$c(\lambda M + (1 - \lambda)M') > \lambda c(M) + (1 - \lambda)c(M')$
for some $(\lambda, (M, M')) \in (0, 1) \times \bM^2$,
then slightly decreasing the cost $c(\lambda M + (1 - \lambda)M')$
does not change the \dm's preference.
To exclude such a case, we focus on cost structures with convexity.
A cost structure $(\bM, c)$ is \emph{convex} if $\bM$ and $c$ are convex.
A costly ambiguity perception representation $\tuple{u, (\bM, c)}$
is \emph{convex} if $(\bM, c)$ is convex.

\providecommand{\pat}{\approx}

Convexity alone still does not pin down the \dm's set of feasible perceptions.
For instance, we can freely exclude an ambiguity perception
that is never strictly optimal from the feasible set.
To explore the smallest possible convex cost structure,
we introduce the following definitions.
For each binary relation $\rel$ on $\lott(\cF)$,
a \emph{multi-{\meu} representation of $\rel$} is a pair $(u, \bM)$ of
a surjective vNM function and a compact convex subset of $\bK$ such that
$P \rel Q$ if and only if
\begin{equation}
  \int \Bigl(\min_{\mu \in M} \int u \cmpf f \diff \mu\Bigr)\diff P(f)
  \ge \int \Bigl(\min_{\mu \in M} \int u \cmpf f \diff \mu\Bigr) \diff Q(f)
  \quad \forall M \in \bM.
\end{equation}
Define the relation $\wpr^*$ on $\lott(\cF)$ by
$P \wpr^* Q$ if $\lambda P + (1 - \lambda)R \wpr \lambda Q + (1 - \lambda)R$
for each $(\lambda, R) \in (0, 1] \times \lott(\cF)$.
For each pair $(u, v)$ of vNM functions, write $u \pat v$
if they coincide up to a positive affine transformation.

\begin{proposition}\label{prop:mmeu}
  \hfill
  \begin{enumerate}
    \item \label{item:mmeu_rep}
      If $\wpr$ is a costly ambiguity perception preference,
      then $\wpr^*$ has a multi-{\meu} representation.
    \item \label{item:mmeu_uniq}
      If $(u, \bM)$ and $(v, \bL)$ are
      multi-{\meu} representations of the same binary relation on $\lott(\cF)$,
      then $u \pat v$ and $\bM = \bL$.
  \end{enumerate}
\end{proposition}

We are now ready to state our identification result.
For each vNM function $u$ representing the restriction of
a costly ambiguity perception preference $\wpr$ to $\Dirac[\lott(X)]$,
define $c^\star_{\wpr, u} \colon \bK \to [0, \infty]$ by
\begin{equation}\label{eq:cn}
  c^\star_{\wpr, u}(M) = \sup_{P \in \lott (\cF)} \Bigl[
    \int \Bigl(\min_{\mu \in M} \int u \cmpf f \diff \mu\Bigr) \diff P(f)
    - u(\bar P)\Bigr],
\end{equation}
where for each $P \in \lott(\cF)$,
let $\bar P \in \lott(X)$ be such that $\Dirac[\bar P] \ipr P$.
The next theorem shows four important properties of convex cost structures:
\begin{enumerate*}
  \item
    every costly ambiguity perception preference has a convex representation;
  \item
    the smallest feasible set of ambiguity perceptions exists;
  \item
    the cost function is unique given a vNM function and a feasible set;
  \item
    the unique cost function can be recovered from observable data.
\end{enumerate*}

\begin{theorem}\label{thm:id}
  Let $\wpr$ be a costly ambiguity perception preference and
  let $(v, \bM^*)$ be a multi-{\meu} representation of $\wpr^*$.
  A pair $\tuple{u, (\bM, c)}$ of a vNM function and a convex cost structure
  is a costly ambiguity perception representation of $\wpr$ if and only if
  $u \pat v$, $\bM^* \subseteq \bM$, and $c = c^\star_{\wpr, u}|_\bM$.
\end{theorem}

According to \cref{thm:id},
if $(v, \bM^*)$ is a multi-{\meu} representation of $\wpr^*$,
then $\bM^*$ is the most parsimonious feasible set
of a costly ambiguity perception representation of $\wpr$.
Thus, any ambiguity perception outside $\bM^*$
can never be a strictly optimal choice in the \dm's maximization.
The parsimonious representation is useful
to see the implications of additional axioms.
Indeed, when characterizing special cases in \cref{subsec:spec},
we restrict the structure of each feasible ambiguity perception
starting from the parsimonious representation.

The proof of \cref{thm:id} is based on an intermediate result,
\cref{prop:nv_var_uniq},
which generalizes Theorem 2 of \citet{DDMO2017} to a more abstract setting.
The uniqueness of the feasible set
does not have a counterpart in \citet{DDMO2017},
since they take as the domain of cost functions the largest possible set.

Note that the richness of the domain due to ex ante randomization
is essential for the uniqueness established in \cref{thm:id}.
The following example shows that
two preferences with different parameters
may exhibit different attitudes
toward ex ante randomization and randomization timing
while coinciding on the set of acts.

\begin{example}\label{ex:id}
  Let $\Omega = \{0, 1\}$ and let $X = \SR$.
  Identify each $\mu \in [0, 1]$
  with the prior such that the probability of state $1$ is $\mu$.
  Define $u \colon \lott(X) \to \SR$ by $u(p) = \int x \diff p(x)$.
  Consider the two costly ambiguity perception representations
  $\tuple{u, (\bM_1, c_1)}$ and $\tuple{u, (\bM_2, c_2)}$,
  where $\bM_1 = \{\{\thalf\}\}$,
  $\bM_2 = \set{[t/2, (1 + t)/2] \mvert t \in [0, 1]}$, and
  $c_1$ and $c_2$ are the zero functions.
  For each $i \in \{1, 2\}$,
  let $\wpr_i$ be the preference represented by $\tuple{u, (\bM_i, c_i)}$ and
  let $U_i$ be the associated utility function.
  The preference $\wpr_1$
  is neutral to ex ante randomization and the timing of randomization
  since it is an expected utility preference,
  whereas $\wpr_2$ is not.
  Indeed, defining acts $f$ and $g$ by
  $f(\omega) = \Dirac[\omega]$ and $g(\omega) = \Dirac[1 - \omega]$ gives
  \begin{equation}
    U_2(\Dirac[f]) = U_2(\Dirac[g]) = \half, \qquad
    U_2\Bigl(\half \Dirac[f] + \half \Dirac[g]\Bigr)
    = \max_{M \in \bM_2} \Bigl[
      \half \min_{\mu \in M} (1 - \mu) + \half \min_{\mu \in M} \mu\Bigr]
    = \frac{1}{4}.
  \end{equation}
  Nevertheless, $\wpr_1$ and $\wpr_2$ coincide on the set of acts:
  for each $h \in \cF$,
  letting $d = u(h(1)) - u(h(0))$, $d^+ = \max \{d, 0\}$, and
  $d^- = \max \{-d, 0\}$ gives
  \begin{align}
    U_2(\Dirac[h])
    &= \max_{M \in \bM_2} \min_{\mu \in M} [u(h(0)) + \mu(d^+ - d^-)]
    = u(h(0)) + \max_{t \in [0, 1]} \Bigl(\frac{t}{2} d^+ - \half d^-\Bigr) \\
    &= u(h(0)) + \half(d^+ - d^-)
    = \half u(h(0)) + \half u(h(1)) = U_1(\Dirac[h]).
  \end{align}
  Therefore, we cannot distinguish $\wpr_1$ and $\wpr_2$
  by observing choices only over acts.
  Note that the indistinguishability can also be
  confirmed from Proposition 5 of \citet{CFIL2022}:
  $\bM_1$ and $\bM_2$ have the same half-space closure.
\end{example}

\section{Comparatives}\label{sec:cmp}

Consider two {\dm}s $1$ and $2$.
Each {\dm} $i$ has a preference $\wpr_i$ on $\lott(\cF)$
with a costly ambiguity perception representation
$\tuple{u_i, (\bM_i, c_i)}$.

\subsection{Ambiguity aversion}\label{subsec:aa}

The following definition extends the notion of comparative ambiguity aversion
by \citet{GhirardatoMarinacci2002} to include ex ante randomization.

\begin{definition}\label{def:cmp_amb}
  The relation $\wpr_1$ is \emph{more ambiguity-averse than $\wpr_2$}
  if for each $(P, p) \in \lott(\cF) \times \lott(X)$,
  \begin{equation}
    \Dirac[p] \spr_2 P \implies \Dirac[p] \spr_1 P.
  \end{equation}
\end{definition}

A {\dm} prefers a constant act to an ambiguous lottery
if the net benefit from her choice of ambiguity perceptions
does not exceed the value of the constant act.
Thus, if {\dm} $1$ prefers a constant act to an ambiguous lottery
more frequently than {\dm} $2$,
then {\dm} $1$ does not obtain a higher net benefit
from tuning her ambiguity perceptions than {\dm} $2$,
which in turn means that
{\dm} $1$ faces a more stringent perception-choice problem.
This difference can be captured through the pointwise ordering
of their cost functions.

\begin{proposition}\label{prop:cmp_amb}
  The relation $\wpr_1$ is more ambiguity-averse than $\wpr_2$ if and only if
  $u_1 \pat u_2$ and $c^\star_{\wpr_1, u_1} \ge c^\star_{\wpr_2, u_1}$.
\end{proposition}

If we assume that $\tuple{u_i, (\bM_i, c_i)}$ is
a $\supseteq$-increasing convex representation for each $i \in \{1, 2\}$,
then the following monotonicity about feasible sets can be obtained
from \cref{thm:id} and \cref{prop:cmp_amb}:
if $\wpr_1$ is more ambiguity-averse than $\wpr_2$,
then $\bM_1 \subseteq \bM_2$.

\subsection{Aversion to additional ambiguity}\label{subsec:aaa}

The following notion strengthens comparative ambiguity aversion.

\begin{definition}\label{def:cmp_inc}
  The relation $\wpr_1$ is
  \emph{more averse to additional ambiguity than $\wpr_2$}
  if for each $(\lambda, (P, Q), p)
  \in [0, 1] \times \lott(\cF)^2 \times \lott(X)$,
  \begin{equation}\label{eq:cmp_inc}
    \lambda \Dirac[p] + (1 - \lambda)Q \spr_2 \lambda P + (1 - \lambda)Q
    \implies
    \lambda \Dirac[p] + (1 - \lambda)Q \spr_1 \lambda P + (1 - \lambda)Q.
  \end{equation}
\end{definition}

This comparison captures the key idea of increases in ambiguity.
Compared to $\lambda \Dirac[p] + (1 - \lambda)Q$,
the lottery $\lambda P + (1 - \lambda)Q$ has additional ambiguity
(e.g., the latter has more relevant ambiguous events
due to the difference between $\Dirac[p]$ and $P$).
For evaluations of constant acts,
the choice of ambiguity perceptions is irrelevant.
Thus, in the lottery $\lambda \Dirac[p] + (1 - \lambda)Q$,
replacing $\Dirac[p]$ with any other lottery $P$
makes the {\dm}s' perception-choice problem more complicated.
If $\wpr_1$ is more averse to additional ambiguity than $\wpr_2$,
then {\dm} $1$ obtains less benefit
from the choice of her ambiguity perception,
which corresponds to evaluating lotteries with a larger ambiguity perception.

To formalize this intuition, for each $i \in \{1, 2\}$,
define the correspondence $\cC_i \colon \lott(\cF) \coto \bM_i$ by
\begin{equation}
  \cC_i(P)
  = \argmax_{M \in \bM_i} \Bigl[
    \int \Bigl(\min_{\mu \in M} \int u_i \cmpf f \diff \mu\Bigr)\diff P(f)
    - c_i(M)\Bigr].
\end{equation}
That is, $\cC_i (P)$ is
the set of {\dm} $i$'s optimal ambiguity perceptions at $P$.
Define the relation $\usd$ on $2^\bK$ by $\bM \usd \bM'$ if
for each $M' \in \bM'$, there exists $M \in \bM$ such that $M \subseteq M'$.
The next proposition formally establishes the above relationship.

\begin{proposition}\label{prop:cmp_inc}
  Suppose that $\tuple{u_i, (\bM_i, c_i)}$ is convex for each $i \in \{1, 2\}$.
  The relation $\wpr_1$ is more averse to additional ambiguity than $\wpr_2$
  if and only if
  $u_1 \pat u_2$ and $\cC_2(P) \usd \cC_1(P)$ for each $P \in \lott(\cF)$.
\end{proposition}

By definition,
if $\wpr_1$ is more averse to additional ambiguity than $\wpr_2$,
then $\wpr_1$ is more ambiguity-averse than $\wpr_2$.
The converse is not true as the following example shows.

\begin{example}\label{ex:cmp}
  Let $\Omega = \{0, 1\}$ and let $X = \SR$.
  Identify each $\mu \in [0, 1]$
  with the prior such that the probability of state $1$ is $\mu$.
  Define $u \colon \lott(X) \to \SR$ by $u(p) = \int x \diff p(x)$.
  Consider the two costly ambiguity perception representations
  $\tuple{u, (\bM_1, c_1)}$ and $\tuple{u, (\bM_2, c_2)}$,
  where $\bM_1 = \bM_2 = \set{[\lambda, 1] \mvert \lambda \in [0, 1]}$ and
  for each $\lambda \in [0, 1]$,
  \begin{equation}
    c_1([\lambda, 1]) = 2\lambda, \qquad
    c_2([\lambda, 1]) = \max \{\lambda, 3\lambda - 1\}.
  \end{equation}
  For each $i \in \{1, 2\}$,
  let $\wpr_i$ be the preference represented by $\tuple{u, (\bM_i, c_i)}$.
  Since $\bM_1$ and $\bM_2$ are $\supseteq$-increasing and since $c_1 \ge c_2$,
  it follows from \cref{prop:cmp_amb} that
  $\wpr_1$ is more ambiguity-averse than $\wpr_2$.
  Since the act $f$ of the form $f(\omega) = \Dirac[2\omega]$ satisfies
  \begin{equation}
    \cC_1(\Dirac[f]) = \bM_1, \qquad
    \cC_2(\Dirac[f]) = \Bigl\{\Bigl[\half, 1\Bigr]\Bigr\},
  \end{equation}
  we have $\cC_2(\Dirac[f]) \not\mathrel{\usd} \cC_1(\Dirac[f])$,
  which shows by \cref{prop:cmp_inc} that
  $\wpr_1$ is not more averse to additional ambiguity than $\wpr_2$.
\end{example}

\subsection{Aversion to ex ante randomization}\label{subsec:aear}

For each $(i, P) \in \{1, 2\} \times \lott(\cF)$,
let $\bar P_i \in \lott(X)$ be such that $\Dirac[\bar P_i] \ipr_i P$.

\begin{definition}\label{def:cmp_eaar}
  The relation $\wpr_1$ is
  \emph{more averse to ex ante randomization than $\wpr_2$}
  if for each $(\lambda, (P, Q)) \in [0, 1] \times \lott (\cF)^2$,
  \begin{equation}
    \lambda \Dirac[\bar P_2] + (1 - \lambda)\Dirac[\bar Q_2]
    \spr_2 \lambda P + (1 - \lambda)Q \implies
    \lambda \Dirac[\bar P_1] + (1 - \lambda)\Dirac[\bar Q_1]
    \spr_1 \lambda P + (1 - \lambda)Q.
  \end{equation}
\end{definition}

Given the other axioms for costly ambiguity perception preferences,
the following condition is equivalent to \axmref{axm:eaar}
(see \cref{lem:aux_eaar} in \cref{subsec:aux}):
for each $(\lambda, (P, Q), (p, q))
\in [0, 1] \times \lott(\cF)^2 \times \lott(X)^2$,
if $P \ipr \Dirac[p]$ and $Q \ipr \Dirac[q]$,
then $\lambda \Dirac[p] + (1 - \lambda)\Dirac[q]
\wpr \lambda P + (1 - \lambda)Q$.
If the {\dm} is averse to the ex ante randomization between $P$ and $Q$,
then she cannot deal with it as randomization between unambiguous alternatives,
that is, she exhibits $\lambda \Dirac[p] + (1 - \lambda)\Dirac[q]
\spr \lambda P + (1 - \lambda)Q$.
Hence, the above condition compares
how often {\dm}s avoid ex ante randomization.

Recall that ex ante randomization is not beneficial for {\dm}s in our model
because the optimal ambiguity perception
can change due to that mixture operation.
This negative effect arises
only when the optimal ambiguity perceptions
associated with the two original lotteries differ.
Thus, if $\wpr_1$ is more averse to ex ante randomization than $\wpr_2$,
then {\dm} $1$ employs different ambiguity perceptions more often
than {\dm} $2$.

The following proposition characterizes
comparative aversion to ex ante randomization.

\begin{proposition}\label{prop:cmp_eaar}
  Suppose that $\tuple{u_i, (\bM_i, c_i)}$ is convex for each $i \in \{1, 2\}$.
  The relation $\wpr_1$ is more averse to ex ante randomization than $\wpr_2$
  if and only if for each $(P, Q) \in \lott(\cF)^2$,
  \begin{equation}
    \cC_2(P) \cap \cC_2(Q) = \emptyset \implies
    \cC_1(P) \cap \cC_1(Q) = \emptyset.
  \end{equation}
\end{proposition}

An empty intersection of $\cC_i(P)$ and $\cC_i(Q)$ indicates
that {\dm} $i$ does not use a common ambiguity perception
when evaluating $P$ and $Q$.

\section{Discussion}\label{sec:disc}

\subsection{Maxmin expected utility}\label{subsec:meu}

In a costly ambiguity perception representation,
if the feasible set is a singleton,
the preference boils down to a natural extension
of \citets{GilboaSchmeidler1989} {\meu} preference to the current framework.
In this case, the {\dm} uses a single fixed ambiguity perception
to evaluate all the lotteries.
Formally, an \emph{expected {\meu} representation of $\wpr$} is
a pair $(u, M)$ of a surjective vNM function and a member of $\bK$ such that
\begin{equation}
  P \mapsto
     \int \Bigl(\min_{\mu \in M} \int u \cmpf f \diff \mu\Bigr) \diff P(f)
\end{equation}
represents $\wpr$.
As discussed in \citet{BLP2011},
the {\meu} model cannot accommodate the typical pattern
in the reflection example (\cref{tab:ref}).

At the axiomatic level,
expected {\meu} preferences differ from costly ambiguity perception preferences
in the attitude toward ex ante randomization.
They are neutral to ex ante randomization rather than averse to it.

\begin{axiom}[Ex ante neutrality to randomization]\label{axm:eanr}
  For each $(\lambda, (P, Q)) \in [0, 1] \times \lott(\cF)^2$,
  if $P \ipr Q$, then $P \ipr \lambda P + (1 - \lambda)Q$.
\end{axiom}

The neutrality excludes aversion to increases in relevant ambiguous events and
is therefore incompatible with the pattern described in the temperature example
in the \nameref{sec:intro}.

\begin{corollary}\label{cor:rep_meu}
  A costly ambiguity perception preference $\wpr$ satisfies \axmref{axm:eanr}
  if and only if it has an expected {\meu} representation.
\end{corollary}

\subsection{Optimization of ambiguity perceptions}\label{subsec:spec}

We discuss three special cases of the costly ambiguity perception model,
in each of which the {\dm} optimally chooses her ambiguity perception.
They have been studied in the literature
under a framework without ex ante randomization (see also \cref{subsec:lit}).

The first two cases restrict the \dm's feasible set of ambiguity perceptions.

A \emph{capacity} is a real-valued function $\nu$ on $2^\Omega$
that satisfies $\nu(\emptyset) = 0$, $\nu(\Omega) = 1$,  and
$\nu(E) \le \nu(F)$ for each $(E, F) \in 2^\Omega \times 2^\Omega$
with $E \subseteq F$.
For each capacity $\nu$,
the \emph{core of $\nu$} is the set $\set{\mu \in \prob(\Omega)
\mvert \mu(E) \ge \nu(E) \tforeach E \in 2^\Omega}$, and
is denoted by $\core(\nu)$.
A capacity $\nu$ is \emph{convex}
if $\nu(E) + \nu(F) \le \nu(E \cap F) + \nu(E \cup F)$
for each $(E, F) \in 2^\Omega \times 2^\Omega$.
Every convex capacity has a nonempty core.
An \emph{optimal ambiguity attitude representation of $\wpr$}
is a costly ambiguity perception representation $\tuple{u, (\bM, c)}$
where each member of $\bM$ is the core of some convex capacity \citep{PTX2026}.
For each $\phi \in \SR^\Omega$, its Choquet integral $\int \phi \diff \nu$
with respect to a convex capacity $\nu$,
defined as $\int \phi \diff \nu
= \int_0^\infty \nu(\phi^{-1}([t, \infty))) \diff t
+ \int_{-\infty}^0 (\nu(\phi^{-1}([t, \infty))) - 1) \diff t$,
equals $\min_{\mu \in \core(\nu)} \int \phi \diff \mu$.
Thus, if $\wpr$ has an optimal ambiguity attitude representation
$\tuple{u, (\bM, c)}$,
then there exists a pair $(N, C)$ of a set of convex capacities and
a grounded function on $N$ such that
\begin{equation}
  P \mapsto \max_{\nu \in N} \Bigl[
    \int \Bigl( \int u \cmpf f \diff \nu\Bigr) \diff P(f) - C(\nu)
    \Bigr]
\end{equation}
represents $\wpr$.

The second case restricts feasible ambiguity perceptions to singletons.
A \emph{moral-hazard representation of $\wpr$}
is a pair $\tuple{u, (M, C)}$ of a surjective vNM function and
a pair of a nonempty compact subset of $\prob(\Omega)$ and
a lower semicontinuous grounded function on $M$ such that
\begin{equation}
  P \mapsto \max_{\mu \in M} \Bigl[
    \int \Bigl( \int u \cmpf f \diff \mu\Bigr) \diff P(f) - C(\mu)
    \Bigr]
\end{equation}
represents $\wpr$ \citep{Sinander2025}.

In the third special case,
the {\dm} chooses her ambiguity perception without costs.
A \emph{dual-self expected utility representation of $\wpr$}
is a pair $(u, \bM)$ of
a surjective vNM function and a nonempty compact subset of $\bK$
such that
\begin{equation}
  P \mapsto
    \max_{M \in \bM} \int \Bigl(\min_{\mu \in M} \int u \cmpf f \diff \mu\Bigr)
    \diff P(f)
\end{equation}
represents $\wpr$ \citep{CFIL2022}.

The restrictions on the feasible sets of ambiguity perceptions
in the first two cases
represent different attitudes toward randomization timing.
The following two axioms place further restrictions than \axmref{axm:irtc} on
when the {\dm} believes ex post randomization provides a hedge.
A pair $(f, g) \in \cF^2$ is \emph{comonotonic}
if no $(\omega, \omega') \in \Omega^2$ satisfies
$\Dirac[f(\omega)] \spr \Dirac[f(\omega')]$ and
$\Dirac[g(\omega')] \spr \Dirac[g(\omega)]$.

\begin{axiom}[Indifference to randomization timing of comonotonic acts]
  \label{axm:irtcm}
  For each $((\kappa, \lambda), P, (f, g))
  \in [0, 1]^2 \times \lott(\cF) \times \cF^2$,
  if $(f, g)$ is comonotonic,
  then $\kappa \Dirac[\lambda f + (1 - \lambda)g] + (1 - \kappa)P
  \ipr \kappa[\lambda \Dirac[f] + (1 - \lambda)\Dirac[g]] + (1 - \kappa)P$.
\end{axiom}

\begin{axiom}[Indifference to randomization timing]\label{axm:irt}
  For each $((\kappa, \lambda), P, (f, g))
  \in [0, 1]^2 \times \lott(\cF) \times \cF^2$,
  we have $\kappa \Dirac[\lambda f + (1 - \lambda)g] + (1 - \kappa)P
  \ipr \kappa[\lambda \Dirac[f] + (1 - \lambda)\Dirac[g]] + (1 - \kappa)P$.
\end{axiom}

\axmref*{axm:irtcm} requires that
the {\dm} believes that mixing two comonotonic acts does not provide a hedge.
If each of her ambiguity perceptions is the core of a convex capacity,
then the associated {\meu} function preserves the mixture of comonotonic acts,
so the ex post randomization of comonotonic acts is not beneficial to her.
Under \axmref{axm:irt}, randomization timing never matters.
If the {\dm} perceives no ambiguity,
then she has no reason to prefer ex post randomization,
as \axmref{axm:aepr} is motivated by a preference
for hedging against ambiguity.

Strengthening \axmref{axm:ica} to the following
characterizes costless choices of ambiguity perceptions.

\begin{axiom}[Strong independence of constant acts]\label{axm:sica}
  For each $(\lambda, (P, Q), p)
  \in (0, 1] \times \lott(\cF)^2 \times \lott(X)$,
  we have $P \wpr Q$ if and only if
  $\lambda P + (1 - \lambda)\Dirac[p] \wpr \lambda Q + (1 - \lambda)\Dirac[p]$.
\end{axiom}

If the \dm's costs are always zero,
then her optimal ambiguity perception at each lottery remains unchanged
when the lottery is ex ante randomized with a constant act.
Thus, ex ante randomization with the same constant act
does not change the ranking between two lotteries.

The following corollary collects the characterizations of the three models.

\begin{corollary}\label{cor:rep_spec}
  Suppose that $\wpr$ is a costly ambiguity perception preference.
  \begin{enumerate}
    \item \label{item:rep_oap}
      The relation $\wpr$ satisfies \axmref{axm:irtcm} if and only if
      it has an optimal ambiguity attitude representation.
    \item \label{item:rep_mh}
      The relation $\wpr$ satisfies \axmref{axm:irt} if and only if
      it has a moral-hazard representation.
    \item \label{item:rep_dseu}
      The relation $\wpr$ satisfies \axmref{axm:sica} if and only if
      it has a dual-self expected utility representation.
  \end{enumerate}
\end{corollary}

Each of the three models is compatible with the typical pattern
in the reflection example (\cref{tab:ref}).
Every moral-hazard preference must produce
(the weak version of) the typical pattern,
since it satisfies the opposite of \textit{uncertainty aversion}
(i.e., the ex post randomization of two indifferent acts can never be better
than the original acts).
For the same reason, it is incompatible with the Ellsberg-type behavior.
We present in \cref{sec:machina} an example of
an absolutely ambiguity-averse dual-self expected utility preference
that generates the typical choices.
An important behavioral pattern
excluded from the dual-self expected utility model, but not from ours,
is the $50$--$51$ example of \citet{Machina2009}
(see \cref{sec:machina} for details).
This example is also known to be difficult to accommodate
in standard ambiguity models \citep{Machina2009,BLP2011}.
Finally, as \citet{PTX2026} point out,
an optimal ambiguity attitude preference can simultaneously
be absolutely ambiguity-averse and
explain both the reflection and $50$--$51$ examples.
Indeed, \cref{ex:ref} falls within
the class of optimal ambiguity attitude models and
generalizes the numerical example presented by \citet{PTX2026}.

\subsection{Intersection of ambiguity perceptions}\label{subsec:int}

The intersection of feasible ambiguity perceptions
in a costly ambiguity perception representation is
the set of priors that the {\dm} cannot eliminate,
whichever perception she chooses.
We show that the intersection is uniquely identified from the preference
by providing an explicit representation.
For each $(\mu, P) \in \prob(\Omega) \times \lott(\cF)$,
define $P^\mu \in \lott(X)$ by
$P^\mu(Z) = \int (\int f(\omega)(Z) \diff \mu(\omega))\diff P(f)$,
which is the overall distribution of outcomes given a prior $\mu$.
The \emph{core of $\wpr$} is the set $\bigcap_{P \in \lott(\cF)}
\set{\mu \in \prob(\Omega) \mvert \Dirac[P^\mu] \wpr P}$ and
is denoted by $\core(\wpr)$.
It is the set of priors under which
the overall outcome distribution of each lottery is preferred
to the original lottery.

\begin{proposition}\label{prop:chara_int}
  Let $\tuple{u, (\bM, c)}$
  be a costly ambiguity perception representation of $\wpr$.
  Then, $\bigcap \bM = \core(\wpr)$.
\end{proposition}

Thus, the intersection of feasible ambiguity perceptions
is independent of the choice of the feasible set,
as it always equals the core of the preference.
This uniqueness does not require convexity of the feasible set,
unlike the uniqueness of the cost structure (\cref{thm:id}).

Moreover, the intersection is tied to the notion of absolute ambiguity aversion
due to \citet{GhirardatoMarinacci2002}.
The relation $\wpr$ is \emph{absolutely ambiguity-averse}
if it is more ambiguity-averse than
a nondegenerate subjective expected utility preference;
that is, a binary relation represented by the function
$P \mapsto \int (\int u \cmpf f \diff \mu) \diff P(f)$
for some vNM function $u$ and a prior $\mu \in \prob(\Omega)$.
Absolute ambiguity aversion captures
the typical pattern in \citets{Ellsberg1961} thought experiments:
preference for a risky bet over ambiguous bets.
Theorem 12 of \citet{GhirardatoMarinacci2002} characterizes
absolute ambiguity aversion of a preference over acts
by the nonemptiness of the core.
An analogous result holds in our framework.

\begin{corollary}\label{cor:absaa}
  Let $\tuple{u, (\bM, c)}$ be
  a costly ambiguity perception representation of $\wpr$.
  Then, $\wpr$ is absolutely ambiguity-averse if and only if
  $\bigcap \bM \ne \emptyset$.
\end{corollary}

We omit the proof of \cref{cor:absaa}
since it is a direct consequence of the following proposition for comparatives
of a costly ambiguity perception preference and an expected {\meu} preference.

\providecommand{\wbp}{\geqslant}
\providecommand{\sbp}{>}

\begin{proposition}\label{prop:cmp_meu}
  Let $\tuple{u, (\bM, c)}$ be
  a costly ambiguity perception representation of $\wpr$ and
  let $(v, L)$ be an expected {\meu} representation
  of a binary relation $\wbp$ on $\lott(\cF)$.
  The following statements are equivalent.
  \begin{equivalent}
    \item \label{item:cmp_meu_meet}
      $u \pat v$ and $\bigcap \bM \supseteq L$.
    \item \label{item:cmp_meu_aa}
      $\wpr$ is more ambiguity-averse than $\wbp$.
    \item \label{item:cmp_meu_maaa}
      $\wpr$ is more averse to additional ambiguity than $\wbp$.
  \end{equivalent}
\end{proposition}

Expected {\meu} preferences are natural benchmarks
for measuring aversion to additional ambiguity
since they are neutral to ex ante randomization.
According to \cref{prop:cmp_meu},
comparative aversion to additional ambiguity
of a costly ambiguity perception preference $\wpr$
against an expected {\meu} preference $\wbp$ actually coincides with
comparative ambiguity aversion of $\wpr$ against $\wbp$,
although the two comparatives are different in general (\cref{ex:cmp}).
They are characterized by the same condition \cref{item:cmp_meu_meet}
on the intersection of the feasible set.

\subsection{Dual representation}\label{subsec:dual}

Our costly ambiguity perception model is also closely related to
the double {\meu} model proposed by \citet{KeZhang2020};
it is a ``dual'' version of the dual-self expected utility model.
The domain of their model is the same as ours:
the original Anscombe--Aumann domain with two-stage randomization.

A \emph{variational {\meu} representation of $\wpr$}
is a pair $\tuple{u, (\bM, c)}$ of
a surjective vNM function on $\lott(X)$ and
a pair of a nonempty compact subset of $\bK$ and
a lower semicontinuous grounded function on $\bM$ such that
$M \subseteq M'$ implies $c(M) \le c(M')$ for each $(M, M') \in \bM^2$ and
\begin{equation}
  P \mapsto \min_{M \in \bM} \Bigl[
    \int \Bigl(\min_{\mu \in M} \int u \cmpf f \diff \mu\Bigr) \diff P(f)
    + c(M)\Bigr]
\end{equation}
represents $\wpr$.
A \emph{double {\meu} representation of $\wpr$} is a pair $(u, \bM)$ of
a surjective vNM function and a nonempty compact subset of $\bK$ such that
\begin{equation}
  P \mapsto
    \min_{M \in \bM} \int \Bigl(\min_{\mu \in M} \int u \cmpf f \diff \mu\Bigr)
    \diff P(f)
\end{equation}
represents $\wpr$.
The former representation generalizes
\citets{KeZhang2020} latter representation by adding the cost term.

\citet{KeZhang2020} motivate the double {\meu} model
by considering the case where ex ante randomization may help
the {\dm} hedge against ambiguity because of the \dm's subjective perception
of the timing of randomization.
The next is a key difference from the costly ambiguity perception model.

\begin{axiom}[Ex ante attraction to randomization]\label{axm:eapr}
  For each $(\lambda, (P, Q)) \in [0, 1] \times \lott(\cF)^2$,
  if $P \wpr Q$, then $\lambda P + (1 - \lambda)Q \wpr Q$.
\end{axiom}

An important implication of \axmref{axm:eapr}, compared to \axmref{axm:eaar},
is the possibility of violations of a form of dynamic consistency
(see also \citet{EGK2016}).
Let $(f, g, h) \in \cF^3$.
Suppose that the {\dm} prefers $f$ to $g$, and $g$ to $h$.
Under \axmref{axm:eapr}, the ranking
$\half \Dirac[g] + \half \Dirac[h] \spr \Dirac[f]$ is possible.
That is, ex ante, the {\dm} prefers randomizing $g$ and $h$ over $f$.
However, after the randomization,
she will see the realized act ($g$ or $h$) worse than $f$,
which is dynamically inconsistent.
In contrast, \axmref{axm:eaar} excludes this pattern:
the {\dm} would never choose to randomize over inferior acts
if a superior act is available.

By adopting \axmref{axm:eapr} instead of \axmref{axm:eaar},
we can replace the ``$\max$'' in the costly ambiguity perception representation
with ``$\min$.''

\begin{corollary}\label{cor:rep_kz}
  Assume \axmref{axm:reg,axm:mon,axm:aepr,axm:irtc,axm:ica,axm:ubd}.
  \begin{enumerate}
    \item \label{item:kz_ccap}
      The relation $\wpr$ satisfies \axmref{axm:eapr} if and only if
      it has a variational {\meu} representation.
    \item \label{item:kz_dmeu}
      The relation $\wpr$ satisfies \axmref{axm:eapr,axm:sica} if and only if
      it has a double {\meu} representation.
  \end{enumerate}
\end{corollary}

We omit the proof of \cref{cor:rep_kz} since the same argument used
in the proof of \cref{thm:rep} and \cref{cor:rep_spec} \cref{item:rep_dseu}
applies by changing ``convexity'' to ``concavity'' and
``$\max$'' to ``$\min$.''

In addition to the set of axioms in \cref{cor:rep_kz},
to characterize the double {\meu} model,
\citet{KeZhang2020} impose the following axiom.

\begin{axiom}[Preference for statewise randomization]\label{axm:psr}
  For each $(\lambda, (f, g)) \in (0, 1) \times \cF^2$,
  \begin{enumerate}
    \item \label{item:psr}
      $\Dirac[f] \wpr \Dirac[g]$ implies
      $\Dirac[\lambda f + (1 - \lambda)g] \wpr \Dirac[g]$;
    \item \label{item:psr_cind}
      $\Dirac[f] \wpr \Dirac[g]$ if and only if
      $\Dirac[\lambda f + (1 - \lambda)p]
      \wpr \Dirac[\lambda g + (1 - \lambda)p]$ for each $p \in \lott(X)$.
  \end{enumerate}
\end{axiom}

That is, the restriction of $\wpr$ to $\Dirac[\cF]$ is required to satisfy
the \textit{uncertainty aversion} and \textit{certainty independence} axioms.
As we can see in \cref{cor:rep_kz} \cref{item:kz_dmeu},
it turns out that \axmref{axm:psr} is redundant under the other axioms.
We can directly check this point at the level of axioms.

\begin{proposition}\label{prop:kz_axm}
  If $\wpr$ satisfies \axmref{axm:reg,axm:aepr,axm:irtc,axm:eapr,axm:sica},
  then it satisfies \axmref{axm:psr}.
\end{proposition}

By a similar reasoning, we can see that the variational {\meu} model satisfies
the \textit{weak certainty independence} axiom of \citet{MMR2006Ecta}.
Thus, it coincides with the variational model \citep{MMR2006Ecta}
on the domain of acts, and
so it cannot account for Machina's examples, as shown in \citet{BLP2011}.

\subsection{Related literature}\label{subsec:lit}

The idea of the costly choice of ambiguity perceptions also appears
in the independent work of \citet{PTX2026}.
They take as the primitive a preference over acts;
that is, the domain does not involve ex ante randomization.
By generalizing the Choquet expected utility model \citep{Schmeidler1989},
they propose a model in which the {\dm} chooses
an optimal capacity from a feasible set at a cost;
that is, preferences that can be represented by the function
$f \mapsto \max_{\nu \in N} (\int u \cmpf f \diff \nu - C(\nu))$,
where $N$ is a set of capacities
and $C$ is a grounded function on $N$.
They show that such a representation $\tuple{u, (N, C)}$ of a preference exists
if and only if the preference has
a representation $\tuple{u, (\tilde N, \tilde C)}$ such that
$\tilde N$ consists only of convex capacities.
Thus, the convexity of each feasible capacity is not behaviorally meaningful.
By contrast, in our framework, it is behaviorally relevant
because having a representation with a set of convex capacities implies
\axmref{axm:aepr}.
The key behavioral implication of their model is
aversion to mixing comonotonic acts,
which is also motivated
by the absence of the hedging power of specific types of randomization.
Compared to their representations with convex capacities,
our model allows any set of priors as a feasible ambiguity perception.
Finally, the difference in the domains plays a crucial role
in the identification of parameters as \cref{ex:id} illustrates.
Using the richer domain,
we obtain a sharper identification result,
which allows us to conduct several comparatives on the parameters.
Thus, we see our and their approaches are complementary to each other:
while theirs is more parsimonious, ours is suitable for identification.

Another model in which the {\dm} optimally chooses her ambiguity perception
is the dual-self expected utility model by \citet{CFIL2022}.
They also consider a preference over acts.
The key difference from the costly ambiguity perception model is that
there is no cost function in their representation.
While we interpret the maximization
as a single \dm's optimal choice of ambiguity perception,
their primary interpretation is an intrapersonal game played by two selves:
pessimism and optimism.
The key axiom in their characterization is \textit{certainty independence}:
for each $(\lambda, (f, g), p) \in [0, 1]\times \cF^2 \times \lott(X)$,
we have $\Dirac[f] \wpr \Dirac[g]$ if and only if
$\Dirac[\lambda f + (1 - \lambda)p] \wpr \Dirac[\lambda g + (1 - \lambda)p]$.
In our framework, this property follows from \axmref{axm:irt,axm:sica}
(cf.~\cref{prop:kz_axm} in \cref{subsec:dual}).
As shown in \cref{sec:machina}, it is \textit{certainty independence} that
generates the incompatibility with the $50$--$51$ example.

As discussed in \citet{CFIL2022},
if $\wpr$ has a dual-self expected utility representation $(u, \bM)$,
then there exists a nonempty compact subset $\bL$ of $\bK$ such that
the restriction of $\wpr$ to $\Dirac[\cF]$ is represented by
$f \mapsto \min_{L \in \bL} \max_{\ell \in L} \int u \cmpf f \diff \ell$.
Thus, the ordering of ``$\max$'' and ``$\min$'' has no behavioral meaning
in their framework.
On the other hand, in our framework,
the ``$\max$'' and ``$\min$'' operations cannot be reversed
because they are tied to attitudes
toward ex ante and ex post randomization, respectively.

The costly ambiguity perception model is not the only theory
that accounts for Machina's two examples.
One approach is to relax the assumption of expected utility
over unambiguous prospects \citep{DillenbergerSegal2015,Wang2022},
while our model maintains it.
Another approach is taken by
the expected uncertain utility model \citep{GulPesendorfer2014} and
the two-stage evaluation model \citep{He2021},
which do not assume any objective randomization in their domain.
In these models, the {\dm} evaluates outcomes differently
depending on whether the outcomes occur on an ambiguous event.
In contrast, the vector expected utility model \citep{Siniscalchi2009}
addresses the examples on the domain of acts with ex post randomization,
while maintaining expected utility over unambiguous acts.
It decomposes the \dm's utility into a baseline expected utility term and
an ``adjustment'' term, which penalizes the payoff variability of an act.
The key axiom, \textit{complementary independence},
requires the independence axiom only for acts
having symmetric payoff deviations with respect to some constant act.
Our model is compatible with \textit{complementary independence} but
does not always satisfy it.
In particular, it has a nontrivial intersection
with the vector expected utility model, and
neither model nests the other.
While the vector expected utility model explains the two examples
through the abstract adjustment term,
our model does so through the optimization process.

The ``dual'' version of the costly ambiguity perception model
discussed in \cref{subsec:dual} is
a generalization of \citets{KeZhang2020} double {\meu} model,
which further nests \citets{Saito2015} representation
as a special case.
In their model, the {\dm} chooses her ambiguity perception
to minimize the {\meu} value without costs.
Their primitive is the same as ours, but
the {\dm} \emph{prefers} ex ante randomization.
It captures the \dm's belief that
even ex ante randomization helps hedge against ambiguity.
As discussed after \cref{thm:rep},
whereas they construct their representation as a ``cautious completion''
of a subrelation of the \dm's preference,
we cannot apply that argument due to the cost term.
Instead, our approach exploits the convexity and translation equivariance
of the utility function more directly.

The attitudes toward ex ante randomization under ambiguity
have not been explicitly studied in the experimental literature.
This is because \axmref{axm:eaar} is a new prediction,
while what has been mainly tested is
\citets{Schmeidler1989} \textit{uncertainty aversion}
(ex post aversion to randomization).
Among the few exceptions, \citet{BHL2022} and \citet{AoyamaHanaki2026}
explicitly distinguish between the two types of randomization
to study whether the random incentive system provides subjects
with a hedge against ambiguity.
Since they focus on the two-color Ellsberg-urn experiment,
the special role of ex ante randomization proposed in this paper
does not apply to their setting.
Examining when our model is empirically reasonable
is left for future experimental research.

\appendix

\section{Preliminaries}\label{sec:prelim}

\subsection{Convex analysis}\label{subsec:cvx}

Let $E$ be a normed space and let $E'$ be its norm dual.
Define the bilinear functional $\dual{\cdot, \cdot}$ on $E \times E'$ by
$\dual{x, x'} = x'(x)$.

For each real-valued function $W$ on a subset $D$ of $E$,
the \emph{convex conjugate of $W$}
is the extended real-valued function $W^*$ on $E'$ of the form
$W^*(x') = \sup_{x \in D} (\dual{x, x'} - W(x))$;
the \emph{subdifferential of $W$}
is the correspondence $\partial W \colon D \coto E'$
of the form $\partial W(x)
= \bigcap_{y \in D} \set{x' \in E' \mvert W(y) \ge W(x) + \dual{y - x, x'}}$;
let $\gr \partial W$ be the graph
$\set{(x, x') \in D \times E' \mvert x' \in \partial W(x)}$ of $\partial W$.

\begin{lemma}\label{lem:cvx_conj}
  Let $W$ be a real-valued function on a subset $D$ of $E$.
  \begin{enumerate}
    \item \label{item:cvx_conj_lsccvx}
      $W^*$ is weak$*$ lower semicontinuous and convex.
    \item \label{item:cvx_conj_ineq}
      $W(x) \ge \dual{x, x'} - W^*(x')$ for each $(x, x') \in D \times E'$.
    \item \label{item:cvx_conj_sd}
      $W(x) = \dual{x, x'} - W^*(x')$ if and only if
      $(x, x') \in \gr \partial W$.
    \item \label{item:cvx_conj_sd_val}
      $\partial W$ has weak$*$ closed convex values.
    \item \label{item:cvx_conj_clgr}
      If $W$ is lower semicontinuous and
      if $B$ is a norm bounded weak$*$ closed subset of $E'$,
      then $\gr \partial W \cap (D \times B)$ is norm $\times$ weak$*$ closed.
  \end{enumerate}
\end{lemma}

\begin{proof}
  \cref{item:cvx_conj_lsccvx}\enspace
  Since $W^*$ is the pointwise supremum of the family
  $(\dual{x, \cdot} - W(x))_{x \in D}$ of weak$*$ continuous affine functions,
  it is weak$*$ lower semicontinuous and convex.

  \cref{item:cvx_conj_ineq}\enspace
  For each $(x, x') \in D \times E'$,
  since $W^*(x') \ge \dual{x, x'} - W(x)$,
  we have $W(x) \ge \dual{x, x'} - W^*(x')$.

  \cref{item:cvx_conj_sd}\enspace
  It follows that $(x, x') \in \gr \partial W$ if and only if
  $W(x) \le \dual{x, x'} - (\dual{y, x'} - W(y))$ for each $y \in D$,
  which is equivalent to $W(x) \le \dual{x, x'} - W^*(x')$.
  The statement follows from part \cref{item:cvx_conj_ineq}.

  \cref{item:cvx_conj_sd_val}\enspace
  For each $x \in D$,
  the set $\partial W(x)$ is the intersection of weak$*$ closed half spaces.

  \cref{item:cvx_conj_clgr}\enspace
  By Theorem 2.4.2 (ix) of \citet{Zalinescu2002}.
\end{proof}

Given any extended real-valued function $\gamma$,
let $\dom \gamma$ be the inverse image of $[-\infty, \infty)$ under $\gamma$.
A convex function $W$ is \emph{proper}
if $\dom W \ne \emptyset$ and never assumes the value $-\infty$.

\begin{lemma}\label{lem:cvx_lip}
  Let $W$ be a Lipschitz continuous proper convex function
  on a convex subset $D$ of $E$.
  \begin{enumerate}
    \item \label{item:cvx_lip_sd}
      $\partial W$ has nonempty values.
    \item \label{item:cvx_lip_lin}
      For each $(x, y) \in D^2$, it follows that
      $\partial W(x) \cap \partial W(y) \ne \emptyset$ if and only if
      $W(\lambda x + (1 - \lambda)y) = \lambda W(x) + (1 - \lambda)W(y)$
      for each $\lambda \in [0, 1]$.
  \end{enumerate}
\end{lemma}

\begin{proof}
  \cref{item:cvx_lip_sd}\enspace
  Apply the same argument
  as in the proof of the Duality Theorem of \citet{Gale1967}.

  \cref{item:cvx_lip_lin}\enspace
  Apply the same argument
  as in the proof of Proposition 1 of \citet{Pennesi2015}.
\end{proof}

A subset $C$ of $E$ is \emph{conic}
if $\alpha C \subseteq C$ for each $\alpha \in \SR_+$.
A \emph{cone} is a conic subset of $E$.
For each subset $S$ of $E$,
let $S^+ = \bigcap_{x \in S} \set{x' \in E' \mvert \dual{x, x'} \ge 0}$ and
let $S^{++} = \bigcap_{x' \in S^+} \set{x \in E \mvert \dual{x, x'} \ge 0}$.
The next is Theorem 1.1.9 of \citet{Zalinescu2002}.

\begin{lemma}\label{lem:cvx_cone}
  The closure of a convex conic subset $S$ of $E$ equals $S^{++}$.
\end{lemma}

Let $W$ be a convex function on a convex subset $D$ of $E$.
For each $(x, v) \in D \times E$ with $x + \eps v \in D$
for some $\eps \in \SR_{++}$,
the \emph{right directional derivative of $W$ in the direction $v$ at $x$}
is the number $\Diff^+_v W(x)$ defined as
\begin{equation}
  \Diff^+_v W(x)
  = \lim_{\lambda \downarrow 0} \frac{W(x + \lambda v) - W(x)}{\lambda}.
\end{equation}
The following comes from Theorem 2.4.9 of \citet{Zalinescu2002}.

\begin{lemma}\label{lem:cvx_dd}
  Let $W$ be a continuous proper convex function on $E$.
  Then, $\Diff^+_v W(x) = \max_{x' \in \partial W(x)} \dual{v, x'}$
  for each $(x, v) \in E^2$.
\end{lemma}

For each real-valued function $W$ on a convex subset $D$ of $E$,
define the relation $\rel^*_W$ on $D$ by
\begin{equation}
  x \rel^*_W y \iff
  W(\lambda x + (1 - \lambda)z) \ge W(\lambda y + (1 - \lambda)z)
  \quad \forall (\lambda, z) \in (0, 1] \times D,
\end{equation}
whose graph $\set{(x, y) \in D^2 \mvert x \rel^*_W y}$
is denoted by $\gr(\rel^*_W)$;
let $C^*_W = \set{\alpha(x - y) \mvert
(\alpha, (x, y)) \in \SR_{++} \times \gr(\rel^*_W)}$.

\begin{lemma}\label{lem:cvx_core}
  Let $W$ be a real-valued function on a convex subset $D$ of $E$.
  \begin{enumerate}
    \item \label{item:cvx_core_pro}
      $\rel^*_W$ is reflexive and transitive.
    \item \label{item:cvx_core_ind}
      $\lambda x + (1 - \lambda)z \rel^*_W \lambda y + (1 - \lambda)z$
      for each $(\lambda, (x, y), z) \in (0, 1] \times \gr(\rel^*_W) \times D$.
    \item \label{item:cvx_core_cone}
      $C^*_W$ is a convex cone.
  \end{enumerate}
\end{lemma}

\begin{proof}
  \cref{item:cvx_core_pro}\enspace
  By definition.

  \cref{item:cvx_core_ind}\enspace
  Choose any $(\lambda, (x, y), z) \in (0, 1] \times \gr(\rel^*_W) \times D$.
  For each $(\kappa, z') \in (0, 1) \times D$,
  letting $z''
  = (1 - \kappa \lambda)^{-1}[\kappa(1 - \lambda)z + (1 - \kappa)z']$
  gives $z'' \in D$ and
  $\kappa[\lambda w + (1 - \lambda)z] + (1 - \kappa)z'
  = \kappa \lambda w + (1 - \kappa \lambda)z''$ for each $w \in D$, so
  $\lambda x + (1 - \lambda)z \rel^*_W \lambda y + (1 - \lambda)z$.

  \cref{item:cvx_core_cone}\enspace
  By part \cref{item:cvx_core_pro}, $C^*_W$ contains $0$, so it is conic.
  For each $((\alpha, (x, y)), (\beta, (\tilde x, \tilde y)))
  \in (\SR_{++} \times \gr(\rel^*_W))^2$,
  letting $\lambda = \alpha/(\alpha + \beta)$ gives
  $\lambda x + (1 - \lambda)\tilde x \rel^*_W \lambda y + (1 - \lambda)\tilde x
  \rel^*_W \lambda y + (1 - \lambda)\tilde y$ by part \cref{item:cvx_core_ind},
  so $\alpha(x - y) + \beta(\tilde x - \tilde y)
  = (\alpha + \beta)\{\lambda x + (1 - \lambda)\tilde x
  - [\lambda y + (1 - \lambda)\tilde y]\} \in C^*_W$.
  Thus, $C^*_W$ is convex.
\end{proof}

\subsection{Niveloids}\label{subsec:nv}

Let $(E, \ge)$ be an Archimedean Riesz space with unit $e$;
that is, $E$ is a lattice under the order $\ge$,
the set $\set{n x \mvert n \in \SN}$ is unbounded above
for each nonzero positive $x \in E$, and
for each $x \in E$, there exists $\alpha \in \SR_{++}$ such that
$\alpha e \ge x \ge -\alpha e$.
Denote by $E_+$ the positive cone of $E$.
For each $x \in E$, the \emph{absolute value of $x$} is $x \join (-x)$ and
denoted by $\abs{x}$;
the \emph{essential supremum of $x$} is the number
$\inf \set{\alpha \in \SR \mvert \alpha e \ge x}$ and denoted by $\esup(x)$.
Endow $E$ with the norm $x \mapsto \esup(\abs{x})$.

A real-valued function $W$ on a subset $D$ of $E$ is
\begin{itemize}
  \item
    \emph{monotone} if $W(x) \ge W(y)$
    for each $(x, y) \in D^2$ with $x \ge y$;
  \item
    \emph{normalized} if $W(t e) = t$ for each $t \in \SR$ with $t e \in D$;
  \item
    \emph{translation equivariant} if $W(x + t e) = W(x) + t$
    for each $(x, t) \in D \times \SR$ with $x + t e \in D$;
  \item
    a \emph{niveloid} if $W(x) - W(y) \le \esup(x - y)$
    for each $(x, y) \in D^2$.
\end{itemize}
By definition, every niveloid is $1$-Lipschitz continuous.

A \emph{tube} is a subset $T$ of $E$ such that
$T + \set{t e \mvert t \in \SR} \subseteq T$.
The next is Proposition 2 of \citet{CMMR2014}.

\begin{lemma}\label{lem:nv_mon_ts}
  \hfill
  \begin{enumerate}
    \item \label{item:nv_mon_ts}
      Every niveloid is monotone and translation equivariant.
    \item \label{item:nv_mon_ts_tube}
      Every monotone translation equivariant real-valued function on a tube
      is a niveloid.
  \end{enumerate}
\end{lemma}

\providecommand{\partialp}{\partial_\piup}

Let $\Delta = \set{x' \in E'_+ \mvert \dual{e, x'} = 1}$,
which is weak$*$ compact and convex.
For each real-valued function $W$ on a subset $D$ of $E$,
define the correspondence $\partialp W \colon D \coto E'$ by
$\partialp W(x) = \partial W(x) \cap \Delta$.

\begin{lemma}\label{lem:nv_cvx}
  Let $W$ be a convex niveloid on a convex subset $D$ of $E$.
  \begin{enumerate}
    \item \label{item:nv_cvx_sd}
      If $D = E$, then $\partial W = \partialp W$.
    \item \label{item:nv_cvx_sdp}
      $\partialp W$ has nonempty weak$*$ compact convex values.
  \end{enumerate}
\end{lemma}

\begin{proof}
  \cref{item:nv_cvx_sd}\enspace
  Suppose $D = E$.
  By \cref{lem:nv_mon_ts} \cref{item:nv_mon_ts},
  $W$ is monotone and translation equivariant.
  Thus, for each $(x, x') \in \gr \partial W$,
  since $\dual{-v, x'} = \dual{x - v, x'} - \dual{x, x'}
  \le W(x - v) - W(x) \le 0$ for each $v \in E_+$,
  we have $x' \in E'_+$;
  since $\alpha \dual{e, x'} = \dual{x + \alpha e, x'} - \dual{x, x'}
  \le W(x + \alpha e) - W(x) = \alpha$ for each $\alpha \in \SR$,
  we have $\dual{e, x'} = 1$.

  \cref{item:nv_cvx_sdp}\enspace
  By Theorem 1 and Proposition 4 of \citet{CMMR2014},
  $W$ has a convex niveloidal extension $\hat W$ to $E$.
  For each $x \in D$,
  since $\partial W(x)$ is weak$*$ closed and convex
  by \cref{lem:cvx_conj} \cref{item:cvx_conj_sd_val},
  $\partialp W(x)$ is weak$*$ compact and convex;
  since $\partial \hat W(x) \ne \emptyset$
  by \cref{lem:cvx_lip} \cref{item:cvx_lip_sd} and
  since $\partialp \hat W(x) = \partial \hat W(x) \subseteq \partial W(x)$
  by part \cref{item:nv_cvx_sd},
  we have $\partialp W(x) \ne \emptyset$.
\end{proof}

For each real-valued function $W$ on a subset $D$ of $E$,
a \emph{variational representation of $W$} is a pair $(\Pi, \gamma)$
of a weak$*$ compact subset of $\Delta$ and
a weak$*$ lower semicontinuous extended real-valued function on $\Pi$ such that
$W(x) = \max_{\pi \in \Pi} (\dual{x, \pi} - \gamma(\pi))$ for each $x \in D$.

\begin{lemma}\label{lem:nv_var}
  Let $W$ be a real-valued function on a convex subset $D$ of $E$.
  \begin{enumerate}
    \item \label{item:nv_var_nv}
      If $W$ has a variational representation,
      then it is convex, monotone, and translation equivariant.
    \item \label{item:nv_var_gr}
      If $t e \in D$ for some $t \in \SR$, if $W$ is normalized, and
      if $(\Pi, \gamma)$ is a variational representation of $W$,
      then $\gamma$ is grounded.
  \end{enumerate}
\end{lemma}

\begin{proof}
  \cref{item:nv_var_nv}\enspace
  If $(\Pi, \gamma)$ is a variational representation of $W$,
  then since $\dual{\cdot, \pi} - \gamma(\pi)$ is
  mixture linear, monotone, and translation equivariant for each $\pi \in \Pi$,
  the function $W$ is convex, monotone, and translation equivariant.

  \cref{item:nv_var_gr}\enspace
  Suppose that $t e \in D$ for some $t \in \SR$, that $W$ is normalized, and
  that $(\Pi, \gamma)$ is a variational representation of $W$.
  Since $t = W(t e) = \max_{\pi \in \Pi} (\dual{t e, \pi} - \gamma(\pi))
  = t - \min \gamma(\Pi)$,
  the function $\gamma$ is grounded.
\end{proof}

For each real-valued function $W$ on a convex subset of $E$,
a \emph{multi-expectation representation of $\rel^*_W$}
is a weak$*$ compact convex subset $\Pi$ of $\Delta$ such that
\begin{equation}
  x \rel^*_W y \iff
  \dual{x, \pi} \ge \dual{y, \pi} \quad \forall \pi \in \Pi.
\end{equation}

\begin{proposition}\label{prop:nv_var_me}
  Let $W$ be a convex niveloid on a convex subset $D$ of $E$.
  Then, $\rel^*_W$ has a multi-expectation representation $\Pi$ such that
  $(\Pi, W^*|_\Pi)$ is a variational representation of $W$.
\end{proposition}

\begin{proof}
  Let $\Pi = (C^*_W)^+ \cap \Delta$.
  Then, $\Pi$ is weak$*$ compact and convex by construction, and
  $W^*|_\Pi$ is weak$*$ lower semicontinuous
  by \cref{lem:cvx_conj} \cref{item:cvx_conj_lsccvx}.
  For each $((x, y), \pi) \in \gr(\rel^*_W) \times \Pi$,
  since $x - y \in C^*_W$ and $\pi \in (C^*_W)^+$,
  we have $\dual{x - y, \pi} \ge 0$.
  We show that $(\Pi, W^*|_\Pi)$ is a variational representation of $W$,
  in which case for each $(x, y) \in D^2$,
  if $\dual{x, \pi} \ge \dual{y, \pi}$ for each $\pi \in \Pi$,
  then $\dual{\lambda x + (1 - \lambda)z, \pi} - W^*(\pi)
  \ge \dual{\lambda y + (1 - \lambda)z, \pi} - W^*(\pi)$
  for each $(\lambda, z, \pi) \in (0, 1] \times D \times \Pi$, so
  $x \rel^*_W y$;
  thus, $\Pi$ is a multi-expectation representation of $\rel^*_W$.

  By \cref{lem:cvx_conj} \cref{item:cvx_conj_ineq,item:cvx_conj_sd},
  it suffices to show that $\Pi$ meets $\partialp W(x)$ for each $x \in D$.
  Seeking a contradiction, suppose otherwise.
  Let $x \in D$ be such that $\partialp W(x) \cap \Pi = \emptyset$,
  in which case $\partialp W(x) \cap (C^*_W)^+ = \emptyset$.
  By \cref{lem:nv_cvx} \cref{item:nv_cvx_sdp},
  we can apply the separation theorem
  \citep[Theorem 5.79]{AliprantisBorder2006} to get
  $(\bar x, c) \in E \times \SR$ such that $\bar x \ne 0$ and
  $\dual{\bar x, \pi} < c < \dual{\bar x, x'}$
  for each $(\pi, x') \in \partialp W(x) \times (C^*_W)^+$.
  Since $(C^*_W)^+$ is conic, we have $\bar x \in (C^*_W)^{++}$ and $c < 0$.
  Since $C^*_W$ is dense in $(C^*_W)^{++}$
  by \cref{lem:cvx_cone} and \cref{lem:cvx_core} \cref{item:cvx_core_cone},
  it is without loss of generality to assume $\bar x \in C^*_W$.
  Since by \cref{lem:cvx_conj} \cref{item:cvx_conj_clgr} and
  Theorem 17.11 of \citet{AliprantisBorder2006},
  the correspondence $\partialp W$ is upper hemicontinuous,
  there exists a neighborhood $U$ of $x$ such that
  $\bigcup_{u \in U} \partialp W(u)
  \subseteq \set{\pi \in \Delta \mvert \dual{\bar x, \pi} < c}$.
  Let $(\alpha, (y, z)) \in \SR_{++} \times \gr(\rel^*_W)$ be such that
  $\bar x = \alpha(y - z)$,
  let $\lambda \in (0, 1)$ be such that $\lambda y + (1 - \lambda)x \in U$,
  let $u = \lambda y + (1 - \lambda)x$, and
  let $\eps = \lambda/\alpha$.
  Then, by \cref{lem:cvx_core}, $u = \lambda y + (1 - \lambda)x
  \rel^*_W \lambda z + (1 - \lambda)x = u - \eps \bar x$, so
  $W(u) \ge W(u - \eps \bar x)$.
  However, $W(u - \eps \bar x) - W(u)
  \ge -\dual{\eps \bar x, \pi} > -\eps c
  > 0$ for each $\pi \in \partialp W(u)$,
  which is a contradiction.
\end{proof}

A variational representation $(\Pi, \gamma)$ is \emph{canonical} if
\begin{enumerate}
  \item
    for each $(\pi, \pi') \in \Pi^2$,
    if $\dual{x, \pi} \ge \dual{x, \pi'}$ for each $x \in D$,
    then $\gamma(\pi) \ge \gamma(\pi')$;
  \item
    $\Pi$ and $\gamma$ are convex.
\end{enumerate}
The next generalizes Theorem 2 of \citet{DDMO2017} to an abstract setting.

\begin{proposition}\label{prop:nv_var_uniq}
  Let $(\Pi, \gamma)$ be a canonical variational representation
  of a convex niveloid $W$ on a convex conic subset $D$ of $E$.
  Then, $\gamma = W^*|_\Pi$.
\end{proposition}

\begin{proof}
  Define $\hat \gamma \colon E' \to (-\infty, \infty]$ and
  $\hat W \colon E \to (-\infty, \infty]$ by
  \begin{equation}
    \hat \gamma(x')
    = \inf \gamma(\Pi \cap (\{x'\} + D^+)), \qquad
    \hat W(x)
    = \sup_{x' \in E'} (\dual{x, x'} - \hat \gamma(x')).
  \end{equation}
  By construction, $\epi \hat \gamma = \epi \gamma - D^+ \times \{0\}$ and
  for each $x \in E$,
  \begin{equation}\label{eq:hatW}
    \hat W(x)
    = \sup_{(\pi, x') \in \Pi \times D^+} (\dual{x, \pi - x'} - \gamma(\pi))
    =
    \begin{dcases*}
      \max_{\pi \in \Pi} (\dual{x, \pi} - \gamma(\pi)) & if $x \in D^{++}$, \\
      \infty                                           & else.
    \end{dcases*}
  \end{equation}
  In particular, $\hat W(x) = W(x)$ for each $x \in D$.
  Since $\gamma$ is canonical,
  the function $\hat \gamma$ is a convex extension of $\gamma$.
  Thus, it remains to show that $\hat \gamma$ is weak$*$ lower semicontinuous,
  which implies
  \begin{equation}
    \gamma(\pi) = \hat \gamma(\pi)
    = \sup_{x \in E} (\dual{x, \pi} - \hat W(x))
    = \sup_{x \in D} (\dual{x, \pi} - W(x))
    = W^*(\pi)
    \quad \forall \pi \in \Pi
  \end{equation}
  by the Fenchel--Moreau theorem, \cref{lem:cvx_cone}, and
  the continuity of $\hat W|_{D^{++}}$.

  Let $\net{(\pi_d, r_d), x'_d}$ be a net in $\epi \gamma \times D^+$ such that
  $\net{\pi_d - x'_d, r_d}$ has a weak$*$ limit
  $(\bar y', \bar r) \in E' \times \SR$.
  We show $(\bar y', \bar r) \in \epi \hat \gamma$.
  By the compactness of $\Pi$, we may assume by passing to a subnet that
  $\net{\pi_d}$ has a limit $\bar \pi \in \Pi$.
  Then, $\net{(\pi_d, r_d), x'_d}$ converges to
  $((\bar \pi, \bar r), \bar \pi - \bar y')$,
  which belongs to $\epi \gamma \times D^+$
  by the lower semicontinuity of $\gamma$ and the closedness of $D^+$.
  Therefore, $(\bar y', \bar r) = (\bar \pi, \bar r) - (\bar \pi - \bar y', 0)
  \in \epi \gamma - D^+ \times \{0\}$.
\end{proof}

\begin{lemma}\label{lem:nv_sd}
  Let $W$ be a convex niveloid on a convex conic subset $D$ of $E$ and
  let $\Pi$ be a subset of $\Delta$ such that
  $(\Pi, W^*|_\Pi)$ is a canonical variational representation of $W$.
  \begin{enumerate}
    \item \label{item:nv_sd_lin}
      For each $(x, y) \in D^2$, it follows that
      $\partial W(x) \cap \partial W(y) \cap \Pi \ne \emptyset$ if and only if
      $W(\lambda x + (1 - \lambda)y) = \lambda W(x) + (1 - \lambda)W(y)$
      for each $\lambda \in [0, 1]$.
    \item \label{item:nv_sd_dd}
      $\Diff^+_v W(x) = \max_{\pi \in \partial W(x) \cap \Pi} \dual{v, \pi}$
      for each $(x, v) \in D^2$.
  \end{enumerate}
\end{lemma}

\begin{proof}
  Define $\hat W \colon E \to \SR$ by
  $\hat W(x) = \max_{\pi \in \Pi} (\dual{x, \pi} - W^*(\pi))$,
  which is a convex niveloidal extension of $W$
  by \cref{lem:nv_mon_ts} \cref{item:nv_mon_ts_tube} and
  \cref{lem:nv_var} \cref{item:nv_var_nv}.
  Define $\gamma \colon \Delta \to (-\infty, \infty]$ by
  $\gamma(\pi) = W^*(\pi)$ if $\pi \in \Pi$ and
  $\gamma(\pi) = \infty$ otherwise.
  Then, $(\Delta, \gamma)$
  is a canonical variational representation of $\hat W$,
  so by \cref{prop:nv_var_uniq}, $\gamma = \hat W^*|_\Delta$.
  Hence, by \cref{lem:cvx_conj} \cref{item:cvx_conj_sd},
  $\partial \hat W(x) \subseteq \Pi$ for each $x \in E$.

  \cref{item:nv_sd_lin}\enspace
  Choose any $(x, y) \in D^2$.
  By \cref{lem:cvx_lip} \cref{item:cvx_lip_lin},
  $\partial W(x) \cap \partial W(y) \cap \Pi \ne \emptyset$ implies
  $W(\lambda x + (1 - \lambda)y) = \lambda W(x) + (1 - \lambda)W(y)$
  for each $\lambda \in [0, 1]$.
  Conversely, suppose that
  $W(\lambda x + (1 - \lambda)y) = \lambda W(x) + (1 - \lambda)W(y)$
  for each $\lambda \in [0, 1]$.
  Then, $\hat W(\lambda x + (1 - \lambda)y)
  = \lambda \hat W(x) + (1 - \lambda)\hat W(y)$ for each $\lambda \in [0, 1]$.
  Thus, again by \cref{lem:cvx_lip} \cref{item:cvx_lip_lin},
  $\partial \hat W(x) \cap \partial \hat W(y) \ne \emptyset$.
  Here, by construction,
  $\partial W(x) \cap \partial W(y) \cap \Pi
  \supseteq \partial \hat W(x) \cap \partial \hat W(y)$.

  \cref{item:nv_sd_dd}\enspace
  Choose any $(x, v) \in D^2$.
  Since $\partial W(x) \supseteq \partial \hat W(x)$ by definition,
  it follows from \cref{lem:cvx_dd} that $\Diff^+_v W(x) = \Diff^+_v \hat W(x)
  = \max_{\pi \in \partial \hat W(x)} \dual{v, \pi}
  \le \max_{\pi \in \partial W(x) \cap \Pi} \dual{v, \pi}$.
  For each $\pi \in \partial W(x) \cap \Pi$,
  since $\lambda \dual{v, \pi} \le W(x + \lambda v) - W(x)$
  for each $\lambda \in \SR_{++}$,
  we have $\dual{v, \pi} \le \Diff^+_v W(x)$.
  Thus, $\Diff^+_v W(x)
  = \max_{\pi \in \partial W(x) \cap \Pi} \dual{v, \pi}$.
\end{proof}

\subsection{Cost structures}\label{subsec:cost}

Given any topological space $Y$,
denote by $\Cb(Y)$ the Banach lattice
of bounded continuous real-valued functions on $Y$,
endowed with the pointwise ordering and the uniform norm;
denote by $\ca(Y)$ the Banach lattice
of signed Borel measures on $Y$ of bounded variation,
endowed with the setwise ordering and the total variation norm;
by the Riesz--Markov--Kakutani representation theorem
\citep[Corollary 14.15]{AliprantisBorder2006},
$\ca(Y)$ can be identified with the norm dual of $\Cb(Y)$
when $Y$ is compact and metrizable.

Let $\Phi = \SR^\Omega$, endowed with the uniform norm.
The constant function $\vone_\Omega$ of value one is an order unit of $\Phi$.
A real-valued function $V$ on $\Phi$ is \emph{positively homogeneous}
if $V(\alpha x) = \alpha V(x)$ for each $(\alpha, x) \in \SR_+ \times \Phi$.

Fix any $\bar \omega \in \Omega$.
Let $\Psi = \set{\psi \in \Phi \mvert \psi(\bar \omega) = 0 \tand
\max_{\omega \in \Omega} \abs{\psi(\omega)} = 1}$ and
let $\bU$ be the closed unit ball in $\Cb(\Psi)$.
The constant function $\vone_\bU$ of value one is an order unit of $\Cb(\bU)$.
The norm of each $\xi \in \Cb(\bU)$ equals $\esup(\abs{\xi})$.
For each $\phi \in \Phi \setminus \set{t \vone_\Omega \mvert t \in \SR}$,
there exists a unique $(\alpha, \psi, t) \in \SR_{++} \times \Psi \times \SR$
such that $\phi = \alpha \psi + t \vone_\Omega$.
Thus, every $v \in \bU$ has
a unique positively homogeneous translation equivariant extension $\hat v$
to $\Phi$, which has the form
$\hat v(\alpha \psi + t \vone_\Omega) = \alpha v(\psi) + t$.
For each $m \in \lott(\Phi)$,
define $m^\vee \in \Cb(\bU)$ by $m^\vee(v) = \int \hat v \diff m$.
Let $\Xi = \set{m^\vee \mvert m \in \lott(\Phi)}$.

\begin{lemma}\label{lem:cs_vee}
  \hfill
  \begin{enumerate}
    \item \label{item:cs_vee_ml}
      $m \mapsto m^\vee$ from $\lott(\Phi)$ to $\Cb(\bU)$ is mixture linear.
    \item \label{item:cs_vee_Dirac}
      $\phi \mapsto \Dirac[\phi]^\vee$ from $\Phi$ to $\Cb(\bU)$
      is positively homogeneous and translation equivariant.
    \item \label{item:cs_vee_rng}
      $\Xi$ is a convex conic tube in $\Cb(\bU)$.
  \end{enumerate}
\end{lemma}

\begin{proof}
  \cref{item:cs_vee_ml}\enspace
  For each $(\lambda, (l, m), v) \in [0, 1] \times \lott(\Phi)^2 \times \bU$,
  we have $[\lambda l^\vee + (1 - \lambda)m^\vee](v)
  = \lambda \int \hat v \diff l + (1 - \lambda)\int \hat v \diff m
  = \int \hat v \diff [\lambda l + (1 - \lambda)m]
  = [\lambda l + (1 - \lambda)m]^\vee(v)$.

  \cref{item:cs_vee_Dirac}\enspace
  By the positive homogeneity and translation equivariance
  of $\hat v$ for each $v \in \bU$.

  \cref{item:cs_vee_rng}\enspace
  By part \cref{item:cs_vee_ml}, the set $\Xi$ is convex;
  since for each $(\alpha, m) \in \SR_{++} \times \lott(\Phi)$,
  letting $l(F) = m(\alpha^{-1}F)$ for each $F \subseteq \Phi$ gives
  $\alpha m^\vee = l^\vee$,
  the set $\Xi$ is conic;
  since for each $(m, t) \in \lott(\Phi) \times \SR$,
  letting $l(F) = m(F - \{t \vone_\Omega\})$ for each $F \subseteq \Phi$ gives
  $m^\vee + t \vone_\bU = l^\vee$,
  the set $\Xi$ is a tube.
\end{proof}

A real-valued function on $\Phi$ is \emph{superlinear}
if it is positively homogeneous and concave.
Let $\bV$ be the set of superlinear niveloids on $\Phi$.
Endow $\bV$ with the metric $\rho$ of the form
\begin{equation}
  \rho(V, V') = \max_{\phi \in \bB} \abs{V(\phi) - V'(\phi)},
\end{equation}
where $\bB$ is the closed unit ball of $\Phi$.

\begin{lemma}\label{lem:cs_bU}
  The function $V \mapsto V|_\Psi$
  is a mixture linear topological embedding from $\bV$ to $\bU$.
\end{lemma}

\begin{proof}
  The mixture linearity holds by construction.
  For each $V \in \bV$, since $V|_\Psi$ is continuous and
  since $V(\psi) \ge V(-\vone_\Omega) = -1$ and
  $V(\psi) \le V(\vone_\Omega) = 1$ for each $\psi \in \Psi$,
  we have $V|_\Psi \in \bU$.
  Let $B = \set{(\alpha, \psi, t) \in \SR_+ \times \Psi \times \SR
  \mvert \alpha \psi + t \vone_\Omega \in \bB}$.
  Then, for each $(\alpha, \psi, t) \in B$,
  since $\abs{t} = \abs{(\alpha \psi + t \vone_\Omega)(\bar \omega)} \le 1$,
  we have $\alpha
  = \max_{\omega \in \Omega} \abs{\alpha \psi(\omega) + t - t}
  \le \max_{\omega \in \Omega} \abs{(\alpha \psi + t \vone_\Omega)(\omega)}
  + \abs{t}
  \le 2$.
  For each $(V, V') \in \bV^2$, since
  \begin{equation}
    \rho(V, V')
    = \max_{(\alpha, \psi, t) \in B}
      \abs{V(\alpha \psi + t \vone_\Omega) - V'(\alpha \psi + t \vone_\Omega)}
    = \max_{(\alpha, \psi, t)\in B} \alpha \abs{V(\psi) - V'(\psi)},
  \end{equation}
  the uniform metric $d$ on $\bU$ satisfies
  $d(V|_\Psi, V'|_\Psi) \le \rho(V, V') \le 2 d(V|_\Psi, V'|_\Psi)$.
  Thus, $V \mapsto V|_\Psi$ is a topological embedding.
\end{proof}

For each $M \in \bK$,
the \emph{support function of $M$} is the real-valued function
$H_M$ on $\Phi$ of the form $H_M(\phi) = \min_{\mu \in M} \dual{\phi, \mu}$.

\begin{lemma}\label{lem:cs_supp}
  \hfill
  \begin{enumerate}
    \item \label{item:cs_supp}
      $M \mapsto H_M$ is a surjective mixture linear isometry
      from $\bK$ to $\bV$.
    \item \label{item:cs_supp_mon}
      For each $(M, M') \in \bK^2$, it follows that $M \subseteq M'$
      if and only if $H_M \ge H_{M'}$.
  \end{enumerate}
\end{lemma}

\begin{proof}
  \cref{item:cs_supp}\enspace
  For each $M \in \bK$, since $H_M$ is superlinear
  by Theorem 7.51 of \citet{AliprantisBorder2006} and
  since $H_M$ is monotone and translation equivariant
  by the monotonicity and translation equivariance of $\dual{\cdot, \mu}$
  for each $\mu \in M$,
  we have $H_M \in \bV$.
  Since for each $(\lambda, (M, M'), \phi)
  \in [0, 1] \times \bK^2 \times \Phi$,
  \begin{equation}
    H_{\lambda M + (1 - \lambda)M'}(\phi)
    = \min_{(\mu, \mu') \in M \times M'}
      [\lambda \dual{\phi, \mu} + (1 - \lambda)\dual{\phi, \mu'}]
    = \lambda H_M(\phi) + (1 - \lambda)H_{M'}(\phi),
  \end{equation}
  the function $M \mapsto H_M$ is mixture linear;
  it is surjective by Theorem 7.52 of \citet{AliprantisBorder2006} and
  \cref{lem:nv_cvx} \cref{item:nv_cvx_sd};
  it is an isometry by Corollary 7.59 of \citet{AliprantisBorder2006}.

  \cref{item:cs_supp_mon}\enspace
  By Theorem 2.4.14 (vi) of \citet{Zalinescu2002}.
\end{proof}

Recall $\Delta = \set{\pi \in \Cb(\bU)'_+ \mvert \dual{\vone_\bU, \pi} = 1}$.
For each $\pi \in \Delta$, define $\pi^\wedge \colon \Phi \to \SR$ by
$\pi^\wedge(\phi) = \dual{\Dirac[\phi]^\vee, \pi}$.
By definition, for each $(M, \pi) \in \bK \times \Delta$
with $H_M = \pi^\wedge$, it follows that
$\int H_M \diff m = \int \pi^\wedge \diff m = \dual{m^\vee, \pi}$
for each $m \in \lott(\Phi)$.
For each subset $\Pi$ of $\Delta$,
let $\Pi^\wedge = \set{\pi^\wedge \mvert \pi \in \Pi}$.
Let $\Delta_\bV = \set{\pi \in \Delta \mvert \pi^\wedge \in \bV}$,
endowed with the weak$*$ topology.

\begin{lemma}\label{lem:cs_wedge}
  The function $\pi \mapsto \pi^\wedge$ from $\Delta_\bV$ to $\bV$
  is surjective, mixture linear, and continuous.
\end{lemma}

\begin{proof}
  The mixture linearity follows by definition.
  Since $\Cb(\bU)'$ can be identified with the space of
  normal signed charges on the Borel algebra on $\bU$
  by the Riesz--Markov--Kakutani representation theorem
  \citep[Theorem 14.10]{AliprantisBorder2006} and
  since $\Dirac[V|_\Psi]^\wedge = V$ for each $V \in \bV$,
  the function $\pi \mapsto \pi^\wedge$ is surjective.
  For the continuity, choose any net $\net{\pi_d}$ in $\Delta_\bV$
  with a limit $\bar \pi \in \Delta_\bV$.
  We show that $\net{\rho(\pi_d^\wedge, \bar \pi^\wedge)}$ converges to $0$.
  Fix any $\eps \in \SR_{++}$.
  Let $F$ be a finite subset of $\bB$
  whose $(\eps/3)$-neighborhood includes $\bB$.
  By the $1$-Lipschitz continuity of each member of $\bV$,
  for each $(\phi, d) \in \bB \times \bD$,
  there exists $\psi \in F$ such that
  \begin{align}
    \abs{\pi_d^\wedge(\phi) - \bar \pi^\wedge(\phi)}
    &\le \abs{\pi_d^\wedge(\phi) - \pi_d^\wedge(\psi)}
      + \abs{\pi_d^\wedge(\psi) - \bar \pi^\wedge(\psi)}
      + \abs{\bar \pi^\wedge(\psi) - \bar \pi^\wedge(\phi)} \\
    &< \abs{\dual{\Dirac[\psi]^\vee, \pi_d - \bar \pi}} + \frac{2}{3}\eps.
  \end{align}
  Thus, since $\net{\pi_d}$ is eventually in the weak$*$ neighborhood
  $\bigcap_{\psi \in F} \set{\tilde \pi \in \Delta_\bV
  \mvert \abs{\dual{\Dirac[\psi]^\vee, \tilde \pi - \bar \pi}} < \eps/3}$
  of $\bar \pi$,
  the net $\net{\rho(\pi_d^\wedge, \bar \pi^\wedge)}$
  is eventually in $[0, \eps)$.
\end{proof}

For each cost structure $(\bM, c)$,
let $\bM^\diamond
= \set{\pi \in \Delta \mvert \pi^\wedge = H_M \tforsome M \in \bM}$ and
define $c^\diamond \colon \bM^\diamond \to (-\infty, \infty]$ by
$c^\diamond(\pi) = c(M)$
for each $(\pi, M) \in \bM^\diamond \times \bM$ with $\pi^\wedge = H_M$.

\begin{lemma}\label{lem:cs_nv}
  Let $(\bM, c)$ be a cost structure and
  let $W$ be the real-valued function on $\Xi$ of the form
  $W(m^\vee) = \max_{M \in \bM} (\int H_M \diff m - c(M))$.
  \begin{enumerate}
    \item \label{item:cs_nv}
      $W$ is a normalized convex niveloid that has a variational representation
      $(\bM^\diamond, c^\diamond)$.
    \item \label{item:cs_nv_cn}
      If $(\bM, c)$ is convex, then $c^\diamond = W^*|_{\bM^\diamond}$.
    \item \label{item:cs_nv_sd}
      If $(\bM, c)$ is convex,
      then $[\argmax_{M \in \bM} (\int H_M \diff m - c(M))]^\diamond
      = \partial W(m^\vee) \cap \bM^\diamond$ for each $m \in \lott(\Phi)$.
  \end{enumerate}
\end{lemma}

\begin{proof}
  \cref{item:cs_nv}\enspace
  By construction, $W(m^\vee)
  = \max_{\pi \in \bM^\diamond} (\dual{m^\vee, \pi} - c^\diamond(\pi))$
  for each $m \in \lott(\Phi)$.
  Thus, by \cref{lem:cs_supp} \cref{item:cs_supp} and \cref{lem:cs_wedge},
  $(\bM^\diamond, c^\diamond)$ is a variational representation of $W$.
  Hence, by \cref{lem:nv_mon_ts} \cref{item:nv_mon_ts_tube},
  \cref{lem:nv_var} \cref{item:nv_var_nv}, and
  \cref{lem:cs_vee} \cref{item:cs_vee_rng}, $W$ is a convex niveloid.
  Since $c$ is grounded, so is $c^\diamond$.
  Therefore, $W$ is normalized.

  \cref{item:cs_nv_cn}\enspace
  Suppose that $(\bM, c)$ is convex.
  By \cref{lem:cs_supp,lem:cs_wedge},
  $\bM^\diamond$ and $c^\diamond$ are convex, and
  for each $(\pi, \tilde \pi) \in (\bM^\diamond)^2$,
  if $\dual{\xi, \pi} \ge \dual{\xi, \tilde \pi}$ for each $\xi \in \Xi$,
  then $c^\diamond(\pi) \ge c^\diamond(\tilde \pi)$.
  Thus, by part \cref{item:cs_nv}, $(\bM^\diamond, c^\diamond)$
  is a canonical variational representation of $W$.
  Hence, by \cref{prop:nv_var_uniq}, $c^\diamond = W^*|_{\bM^\diamond}$.

  \cref{item:cs_nv_sd}\enspace
  If $(\bM, c)$ is convex, then
  for each $m \in \lott(\Phi)$, by part \cref{item:cs_nv_cn} and
  \cref{lem:cvx_conj} \cref{item:cvx_conj_ineq,item:cvx_conj_sd},
  $[\argmax_{M \in \bM} (\int H_M \diff m - c(M))]^\diamond
  = \argmax_{\pi \in \bM^\diamond} (\dual{m^\vee, \pi} - c^\diamond(\pi))
  = \argmax_{\pi \in \bM^\diamond} (\dual{m^\vee, \pi} - W^*(\pi))
  = \partial W(m^\vee) \cap \bM^\diamond$.
\end{proof}

A binary relation $\rel$ on $\Xi$ is
\begin{itemize}
  \item
    \emph{$\Phi$-monotone}
    if $\phi \ge \psi$ implies $\Dirac[\phi]^\vee \rel \Dirac[\psi]^\vee$
    for each $(\phi, \psi) \in \Phi^2$;
  \item
    \emph{attracted to ex post randomization}
    if $\Dirac[\lambda \phi + (1 - \lambda)\psi]^\vee
    \rel [\lambda \Dirac[\phi] + (1 - \lambda)\Dirac[\psi]]^\vee$
    for each $(\lambda, (\phi, \psi)) \in [0, 1] \times \Phi^2$.
\end{itemize}

\begin{lemma}\label{lem:cs_nv_rep}
  Let $W$ be a normalized convex niveloid on $\Xi$ such that
  $\rel^*_W$ is $\Phi$-monotone and attracted to ex post randomization and
  let $c$ be the extended real-valued function on $\bK$ of the form
  $c(M) = \sup_{m \in \lott(\Phi)} (\int H_M \diff m - W(m^\vee))$.
  Then, there exists a compact convex subset $\bM^*$ of $\bK$ such that
  \begin{enumerate}
    \item \label{item:cs_nv_rep_mmeu}
      for each $(l, m) \in \lott(\Phi)^2$,
      it follows that $l^\vee \rel^*_W m^\vee$ if and only if
      $\int H_M \diff l \ge \int H_M \diff m$ for each $M \in \bM^*$\textup{;}
    \item \label{item:cs_nv_rep}
      for each compact convex subset $\bM$ of $\bK$
      with $\bM^* \subseteq \bM$,
      the pair $(\bM, c|_\bM)$ is a convex cost structure and
      $W(m^\vee) = \max_{M \in \bM} (\int H_M \diff m - c(M))$
      for each $m \in \lott(\Phi)$.
  \end{enumerate}
\end{lemma}

\begin{proof}
  By construction, $c(M) = W^*(\pi)$
  for each $(M, \pi) \in \bK \times \bK^\diamond$ with $H_M = \pi^\wedge$.
  By \cref{prop:nv_var_me},
  $\rel^*_W$ has a multi-expectation representation $\Pi$ such that
  $(\Pi, W^*|_\Pi)$ is a variational representation of $W$.
  Since each member of $\Pi^\wedge$ is positively homogeneous and
  translation equivariant by \cref{lem:cs_vee} \cref{item:cs_vee_Dirac} and
  is monotone and concave by the $\Phi$-monotonicity and
  attraction to ex post randomization of $\rel^*_W$,
  we have $\Pi \subseteq \Delta_\bV$.
  Let $\bM^* = \set{M \in \bK \mvert H_M \in \Pi^\wedge}$.
  By \cref{lem:cs_supp} \cref{item:cs_supp} and \cref{lem:cs_wedge},
  $\bM^*$ is nonempty, compact, and convex.
  Since for each $\pi \in \Pi$,
  there exists $M \in \bM^*$ such that $H_M = \pi^\wedge$,
  we have \cref{item:cs_nv_rep_mmeu} and
  $W(m^\vee) = \max_{M \in \bM^*} (\int H_M \diff m - c(M))$
  for each $m \in \lott(\Phi)$.

  To see \cref{item:cs_nv_rep},
  choose any compact convex subset $\bM$ of $\bK$
  with $\bM^* \subseteq \bM$.
  By \cref{lem:cs_supp} \cref{item:cs_supp_mon},
  we have $c(M) \ge c(M')$
  for each $(M, M') \in \bM^2$ with $M \subseteq M'$.
  For each $m \in \lott(\Phi)$,
  we have $W(m^\vee) \le \max_{M \in \bM} (\int H_M \diff m - c(M))$;
  since $c(M) \ge \int H_M \diff m - W(m^\vee)$ for each $M \in \bM$,
  we have $W(m^\vee) \ge \max_{M \in \bM} (\int H_M \diff m - c(M))$.
  Since $M \mapsto \int H_M \diff m - W(m^\vee)$
  is mixture linear and continuous for each $m \in \lott(\Phi)$
  by \cref{lem:cs_supp} \cref{item:cs_supp},
  the function $c|_\bM$ is lower semicontinuous and convex.
  Since $W^*|_\Pi$ is grounded by \cref{lem:nv_var} \cref{item:nv_var_gr},
  so is $c|_\bM$.
  Thus, $(\bM, c|_\bM)$ is a convex cost structure.
\end{proof}

\section{Proofs}

\subsection{Auxiliary lemmas}\label{subsec:aux}

With an abuse of notation, identify $\lott(X)$ with the set of constant acts.

\begin{lemma}\label{lem:aux_ica}
  Assume \axmref{axm:reg,axm:ica}.
  \begin{enumerate}
    \item \label{item:aux_ica_ind}
      For each $(\lambda, (P, Q, R)) \in (0, 1] \times \lott(\lott(X))^3$,
      we have $P \wpr Q$ if and only if
      $\lambda P + (1 - \lambda)R \wpr \lambda Q + (1 - \lambda)R$.
    \item \label{item:aux_ica_ubd}
      If $\wpr$ satisfies \axmref{axm:ubd},
      then for each $(\lambda, (p, q)) \in (0, 1) \times \lott(X)^2$
      with $\Dirac[p] \spr \Dirac[q]$,
      there exists $R \in \lott(\lott(X))$ such that
      $\lambda R + (1 - \lambda)\Dirac[q] \ipr \Dirac[p]$.
  \end{enumerate}
\end{lemma}

\begin{proof}
  \cref{item:aux_ica_ind}\enspace
  Choose any $(\lambda, (P, Q, R)) \in (0, 1] \times \lott(\lott(X))^3$.
  For each $n \in \SN$, by \axmref{axm:ica},
  $n^{-1}P + (1 - n^{-1})Q \wpr Q$ if and only if
  $k n^{-1}P + (1 - k n^{-1})Q \wpr (k - 1)n^{-1}P + [1 - (k - 1)n^{-1}]Q$
  for each $k \in \{1, \dots, n\}$.
  Thus, by \axmref{axm:reg}, $P \wpr Q$ if and only if
  $\lambda P + (1 - \lambda)Q \wpr Q$, which is equivalent to
  $\lambda P + (1 - \lambda)R \wpr \lambda Q + (1 - \lambda)R$
  by \axmref{axm:ica}.

  \cref{item:aux_ica_ubd}\enspace
  Assume \axmref{axm:ubd}.
  Choose any $(\lambda, (p, q)) \in (0, 1) \times \lott(X)^2$
  with $\Dirac[p] \spr \Dirac[q]$.
  Let $r_0 = p$.
  For each $n \in \SN$, if $\Dirac[r_{n - 1}] \wpr \Dirac[p] \spr \Dirac[q]$,
  then by \axmref{axm:ubd}, we can take $r_n \in \lott(X)$ such that
  $\half \Dirac[r_n] + \half \Dirac[q] \wpr \Dirac[r_{n - 1}]$,
  in which case $\Dirac[r_n] \spr \Dirac[r_{n - 1}]$
  by part \cref{item:aux_ica_ind},
  so $\Dirac[r_n] \wpr \Dirac[p] \spr \Dirac[q]$ .
  Let $n \in \SN$ be such that $2^{-n} < \lambda$.
  Then, $\lambda \Dirac[r_n] + (1 - \lambda)\Dirac[q]
  \spr 2^{-n} \Dirac[r_n] + (1 - 2^{-n})\Dirac[q]
  \wpr \Dirac[p] \spr \Dirac[q]$.
  By the mixture continuity of $\wpr$,
  there exists $\kappa \in (0, 1)$ such that
  $\lambda[\kappa \Dirac[r_n] + (1 - \kappa)\Dirac[q]]
  + (1 - \lambda)\Dirac[q] \ipr \Dirac[p]$.
\end{proof}

\begin{lemma}\label{lem:aux_vnm_surj}
  Assume \axmref{axm:reg,axm:mon,axm:irtc,axm:ica}.
  The relation $\wpr$ satisfies \axmref{axm:ubd} if and only if
  its restriction to $\Dirac[\lott(X)]$ is represented
  by a surjective vNM function.
\end{lemma}

\begin{proof}
  For each $(p, q, r) \in \lott(X)^3$,
  it follows from \axmref{axm:irtc} and
  \cref{lem:aux_ica} \cref{item:aux_ica_ind} that
  $\Dirac[p] \ipr \Dirac[q]$ if and only if
  $\Dirac[\half p + \half r] \ipr \half \Dirac[p] + \half \Dirac[r]
  \ipr \half \Dirac[q] + \half \Dirac[r] \ipr \Dirac[\half q + \half r]$.
  Thus, by the mixture space theorem \citep{HersteinMilnor1953},
  there exists a vNM function $u$ that represents
  the restriction of $\wpr$ to $\Dirac[\lott(X)]$.
  By nondegeneracy and \axmref{axm:mon}, $u$ is nonconstant.
  Thus, by Lemma 59 of \citet{CMMM2011},
  $\wpr$ satisfies \axmref{axm:ubd} if and only if $u$ is surjective.
\end{proof}

The relation $\wpr$ \emph{exhibits increasing desire for constant acts}
if for each $(\lambda, (P, Q), (p, q))
\in [0, 1] \times \lott(\cF)^2 \times \lott(X)^2$
with $P \ipr \Dirac[p]$ and $Q \ipr \Dirac[q]$,
we have $\lambda \Dirac[p] + (1 - \lambda)\Dirac[q]
\wpr \lambda P + (1 - \lambda)Q$.

\begin{lemma}\label{lem:aux_eaar}
  Assume \axmref{axm:reg,axm:ica,axm:ubd}.
  \begin{enumerate}
    \item \label{item:aux_eaar_suff}
      If $\wpr$ satisfies \axmref{axm:eaar},
      then it exhibits increasing desire for constant acts.
    \item \label{item:aux_eaar_nec}
      If $\wpr$ satisfies \axmref{axm:mon,axm:irtc}, and
      if it exhibits increasing desire for constant acts,
      then it satisfies \axmref{axm:eaar}.
  \end{enumerate}
\end{lemma}

\begin{proof}
  \cref{item:aux_eaar_suff}\enspace
  Assume \axmref{axm:eaar}.
  Let $(\lambda, (P, Q), (p, q))
  \in [0, 1] \times \lott(\cF)^2 \times \lott(X)^2$ be such that
  $P \ipr \Dirac[p]$ and $Q \ipr \Dirac[q]$.
  Without loss of generality, assume $P \wpr Q$.
  If $P \ipr Q$,
  then $\Dirac[p] \ipr \lambda \Dirac[p] + (1 - \lambda)\Dirac[q]$
  by \cref{lem:aux_ica} \cref{item:aux_ica_ind} and
  $\Dirac[p] \ipr P \wpr \lambda P + (1 - \lambda)Q$ by \axmref{axm:eaar},
  so $\lambda \Dirac[p] + (1 - \lambda)\Dirac[q]
  \wpr \lambda P + (1 - \lambda)Q$.
  Suppose $P \spr Q$.
  By the mixture continuity of $\wpr$,
  it suffices to show that for each $\eps \in (0, 1)$,
  \begin{equation}\label{eq:pref}
    \lambda \Dirac[p] + (1 - \lambda)\Dirac[q]
    \wpr \eps \bigl[\lambda \Dirac[p] + (1 - \lambda)\Dirac[q]\bigr]
      + (1 - \eps)[\lambda P + (1 - \lambda)Q].
  \end{equation}
  Choose any $\eps \in (0, 1)$.
  By \axmref{axm:eaar}, $\Dirac[p] \wpr \eps \Dirac[p] + (1 - \eps)P$ and
  $\eps \Dirac[q] + (1 - \eps)\Dirac[q] = \Dirac[q]
  \wpr \eps \Dirac[q] + (1 - \eps)Q$.
  By \cref{lem:aux_ica} \cref{item:aux_ica_ubd},
  there exists $R \in \lott(\lott(X))$ such that
  $\eps R + (1 - \eps)\Dirac[q] \ipr \Dirac[p]$.
  By \axmref{axm:ica},
  $\Dirac[p] \ipr \eps R + (1 - \eps)\Dirac[q] \wpr \eps R + (1 - \eps)Q$.
  Thus, \axmref{axm:eaar} implies
  $\Dirac[p] \wpr \lambda[\eps \Dirac[p] + (1 - \eps)P]
  + (1 - \lambda)[\eps R + (1 - \eps)Q]$.
  Since $\Dirac[p] \ipr \lambda \Dirac[p]
  + (1 - \lambda)[\eps R + (1 - \eps)\Dirac[q]]$
  by \cref{lem:aux_ica} \cref{item:aux_ica_ind},
  we have
  \begin{equation}
    \lambda \Dirac[p] + (1 - \lambda)\bigl[\eps R + (1 - \eps)\Dirac[q]\bigr]
    \wpr \lambda\bigl[\eps \Dirac[p] + (1 - \eps)P\bigr]
    + (1 - \lambda)[\eps R + (1 - \eps)Q]
  \end{equation}
  Applying \axmref{axm:ica} yields \eqref{eq:pref}.

  \cref{item:aux_eaar_nec}\enspace
  Choose any $(\lambda, (P, Q)) \in [0, 1] \times \lott(\cF)^2$
  with $P \wpr Q$.
  Under \axmref{axm:mon,axm:irtc}, by the mixture continuity of $\wpr$,
  there exists $(p, q) \in \lott(X)^2$ such that
  $\Dirac[p] \ipr P$ and $\Dirac[q] \ipr Q$;
  thus, if $\wpr$ exhibits increasing desire for constant acts,
  then $P \ipr \Dirac[p] \wpr \lambda \Dirac[p] + (1 - \lambda)\Dirac[q]
  \wpr \lambda P + (1 - \lambda)Q$
  by \cref{lem:aux_ica} \cref{item:aux_ica_ind}.
\end{proof}

\providecommand{\wpru}{\wpr^\mathrm{u}}
\providecommand{\ipru}{\ipr^\mathrm{u}}
\providecommand{\spru}{\spr^\mathrm{u}}

Let $\wpru$ be a binary relation on $\lott(\Phi)$.
Denote by $\ipru$ and $\spru$
the symmetric and asymmetric parts of $\wpru$, respectively.
The relation $\wpru$ is \emph{indifferent to mixture timing of constants}
if $\kappa \Dirac[\lambda \phi + (1 - \lambda)t\vone_\Omega]
+ (1 - \kappa)m
\ipru \kappa[\lambda \Dirac[\phi] + (1 - \lambda)\Dirac[t \vone_\Omega]]
+ (1 - \kappa)m$
for each $((\kappa, \lambda), m, \phi, t)
\in [0, 1]^2 \times \lott(\Phi) \times \Phi \times \SR$.

\begin{lemma}\label{lem:aux_imtc}
  If $\wpru$ is transitive and indifferent to mixture timing of constants, then
  $l^\vee = m^\vee$ implies $l \ipru m$ for each $(l, m) \in \lott(\Phi)^2$.
\end{lemma}

\begin{proof}
  Suppose that
  $\wpru$ is transitive and indifferent to mixture timing of constants.
  Choose any $(m_1, m_2) \in \lott(\Phi)^2$ with $m_1^\vee = m_2^\vee$.
  Then, there exists a triple $(\Theta, \alpha, k)$ of
  a nonempty finite subset of $\Psi$,
  a nonnegative-valued function on $\Theta$, and
  a real number such that
  $m_1^\vee(v) = m_2^\vee(v) = \sum_{\psi \in \Theta} \alpha(\psi)v(\psi) + k$
  for each $v \in \bU$.
  Let $N = \abs{\Theta} + 1$.
  We show that $m_i \ipru N^{-1}(\sum_{\psi \in \Theta}
  \Dirac[N \alpha(\psi) \psi] + \Dirac[N k \vone_\Omega])$
  for each $i \in \{1, 2\}$.
  Fix any $i \in \{1, 2\}$.
  For each $((\kappa, \lambda), l, (a, b), \psi, (r, s))
  \in [0, 1]^2 \times \lott(\Phi) \times \SR_+^2 \times \Psi \times \SR^2$,
  letting $(c, t, d)
  = (\lambda a + (1 - \lambda)b, \lambda r + (1 - \lambda)s,
  \max \{a, b\} + 1)$ gives
  \begin{align}
    \bigl\{\kappa\bigl[\lambda \Dirac[a \psi + r \vone_\Omega]
      + (1 - \lambda)\Dirac[b \psi + s \vone_\Omega]\bigr]
      + (1 - \kappa)l\bigr\}^\vee
    = \bigl[\kappa \Dirac[c\psi + t \vone_\Omega] + (1 - \kappa)l\bigr]^\vee
  \end{align}
  by \cref{lem:cs_vee} \cref{item:cs_vee_ml,item:cs_vee_Dirac}, and
  \begin{align}
    &\kappa\bigl[\lambda \Dirac[a \psi + r \vone_\Omega]
      + (1 - \lambda)\Dirac[b \psi + s \vone_\Omega]\bigr]
      + (1 - \kappa)l \\
    &\quad
    = \kappa\biggl[\lambda \Dirac\Bigl[\frac{a}{d}(d \psi)
      + \Bigl(1 - \frac{a}{d}\Bigr)\frac{d r}{d - a}\vone_\Omega\Bigr]
      + (1 - \lambda)\Dirac\Bigl[\frac{b}{d}(d \psi)
      + \Bigl(1 - \frac{b}{d}\Bigr)\frac{d s}{d - b}\vone_\Omega\Bigr]\biggr]
      + (1 - \kappa)l \\
    &\quad
    \ipru \kappa\Biggl\{\frac{c}{d}\Dirac[d \psi]
      + \Bigl(1 - \frac{c}{d}\Bigr)\biggl[\frac{\lambda(d - a)}{d - c}
      \Dirac\Bigl[\frac{d r}{d - a}\vone_\Omega\Bigr]
      + \frac{(1 - \lambda)(d - b)}{d - c}
      \Dirac\Bigl[\frac{d s}{d - b}\vone_\Omega\Bigr]\biggr]\Biggr\}
      + (1 - \kappa)l \\
    &\quad
    \ipru \kappa\biggl[\frac{c}{d}\Dirac[d \psi]
      + \Bigl(1 - \frac{c}{d}\Bigr)
      \Dirac\Bigl[\frac{d t}{d - c}\vone_\Omega\Bigr]\biggr]
      + (1 - \kappa)l
    \ipru \kappa \Dirac[c \psi + t \vone_\Omega] + (1 - \kappa)l
  \end{align}
  by indifference to mixture timing of constants.
  Thus, there exists $(\sigma_i, \beta_i, \tau_i)
  \in \SR_+^\Theta \times \SR_+^\Theta \times \SR^\Theta$ such that
  $\sum_{\psi \in \Theta} \sigma_i(\psi) = 1$,
  $m_i \ipru \sum_{\psi \in \Theta}
  \sigma_i(\psi) \Dirac[\beta_i(\psi)\psi + \tau_i(\psi)\vone_\Omega]$, and
  $m_i^\vee = [\sum_{\psi \in \Theta}
  \sigma_i(\psi) \Dirac[\beta_i(\psi)\psi + \tau_i(\psi)\vone_\Omega]]^\vee$.
  Since by \cref{lem:cs_vee} \cref{item:cs_vee_ml,item:cs_vee_Dirac},
  \begin{align}
    \sum_{\psi \in \Theta} \alpha(\psi) \Dirac[\psi]^\vee + k \vone_\bU
    &= m_i^\vee
    = \Bigl[\sum_{\psi \in \Theta}
      \sigma_i(\psi) \Dirac[\beta_i(\psi)\psi + \tau_i(\psi)\vone_\Omega]
      \Bigr]^\vee \\
    &= \sum_{\psi \in \Theta} \sigma_i(\psi) \beta_i(\psi) \Dirac[\psi]^\vee
      + \sum_{\psi \in \Theta} \sigma_i(\psi) \tau_i(\psi) \vone_\bU,
  \end{align}
  we have $\alpha(\psi) = \sigma_i(\psi) \beta_i(\psi)$
  for each $\psi \in \Theta$ and
  $k = \sum_{\psi \in \Theta} \sigma_i(\psi) \tau_i(\psi)$.
  Hence, by the indifference to mixture timing of constants of $\wpru$,
  \begin{align}
    m_i
    &\ipru \sum_{\psi \in \Theta}
      \sigma_i(\psi) \Dirac[\beta_i(\psi)\psi + \tau_i(\psi)\vone_\Omega] \\
    &= \sum_{\psi \in \Theta} \sigma_i(\psi) \Dirac\Bigl[
      \frac{1}{N} N \beta_i(\psi)\psi + \Bigl(1 - \frac{1}{N}\Bigr)
      \frac{N}{N - 1} \tau_i(\psi)\vone_\Omega\Bigr] \\
    &\ipru \sum_{\psi \in \Theta} \sigma_i(\psi) \biggl(\frac{1}{N}
      \Dirac[N \beta_i(\psi)\psi]
      + \frac{N - 1}{N}
      \Dirac\Bigl[\frac{N}{N - 1}\tau_i(\psi)\vone_\Omega\Bigr]\biggr) \\
    &= \frac{1}{N}
      \sum_{\psi \in \Theta} \sigma_i(\psi) \biggl(\Dirac[N \beta_i(\psi)\psi]
      + (N - 1)
      \Dirac\Bigl[\frac{1}{N - 1} N \tau_i(\psi)\vone_\Omega
      + \Bigl(1 - \frac{1}{N - 1}\Bigr) 0\Bigr]\biggr) \\
    &\ipru \frac{1}{N}
      \sum_{\psi \in \Theta} \sigma_i(\psi) \Bigl(\Dirac[N \beta_i(\psi)\psi]
      + \Dirac[N \tau_i(\psi)\vone_\Omega]
      + (N - 2) \Dirac[0]\Bigr) \\
    &= \frac{1}{N} \biggl\{
      \sum_{\psi \in \Theta} \bigl[\sigma_i(\psi) \Dirac[N \beta_i(\psi)\psi]
      + (1 - \sigma_i(\psi))\Dirac[0]\bigr]
      + \sum_{\psi \in \Theta}
      \sigma_i(\psi) \Dirac[N \tau_i(\psi)\vone_\Omega] \biggr\} \\
    &\ipru \frac{1}{N} \biggl(
      \sum_{\psi \in \Theta} \Dirac[N \sigma_i(\psi) \beta_i(\psi) \psi]
      + \Dirac\Bigl[N \sum_{\psi \in \Theta}
      \sigma_i(\psi) \tau_i(\psi) \vone_\Omega\Bigr]\biggr) \\
    &= \frac{1}{N} \Bigl(\sum_{\psi \in \Theta} \Dirac[N \alpha(\psi) \psi]
      + \Dirac[N k \vone_\Omega]\Bigr).
    \qedhere
  \end{align}
\end{proof}

The relation $\wpru$ is
\begin{itemize}
  \item
    \emph{monotone} if $\phi \ge \psi$ implies
    $\lambda \Dirac[\phi] + (1 - \lambda)m
    \wpru \lambda \Dirac[\psi] + (1 - \lambda)m$
    for each $(\lambda, m, (\phi, \psi))
    \in [0, 1] \times \lott(\Phi) \times \Phi^2$;
  \item
    \emph{$\Xi$-monotone} if $l^\vee \ge m^\vee$ implies $l \wpru m$
    for each $(l, m) \in \lott(\Phi)^2$.
\end{itemize}

\begin{lemma}\label{lem:aux_mon}
  If $\wpru$ is transitive, monotone, and
  indifferent to mixture timing of constants, then it is $\Xi$-monotone.
\end{lemma}

\begin{proof}
  Suppose that $\wpru$ is transitive, monotone, and
  indifferent to mixture timing of constants.
  Choose any $(m_1, m_2) \in \lott(\Phi)^2$ with $m_1^\vee \ge m_2^\vee$.
  Let $\tuple{\Theta, (l_1, l_2), ((\beta_i, \gamma_i))_{i \in \{1, 2\}}}$ be
  a triple of a nonempty finite subset of $\Psi$,
  a pair of members of $\lott(\Phi)$, and
  a family of members of $\SR_+^\Theta \times \SR$ such that
  for each $i \in \{1, 2\}$,
  \begin{equation}\label{eq:aux_imt_rep}
    l_i = \frac{1}{\abs{\Theta} + 1}\Bigl(
      \sum_{\psi \in \Theta} \Dirac[\beta_i(\psi)\psi]
      + \Dirac[\gamma_i \vone_\Omega]\Bigr), \qquad
    l_i^\vee = m_i^\vee.
  \end{equation}
  Define $\tau \colon \Theta \to \SR$ by
  $\tau(\psi) = \abs{\beta_1(\psi) - \beta_2(\psi)}$.
  Let
  \begin{equation}
    k_1
    = \frac{1}{\abs{\Theta} + 1}\biggl[
      \sum_{\psi \in \Theta} \Dirac[\beta_1(\psi)\psi + \tau(\psi)\vone_\Omega]
      + \Dirac\Bigl[\Bigl(\gamma_1 - \sum_{\psi \in \Theta} \tau(\psi)\Bigr)
      \vone_\Omega\Bigr]\biggr].
  \end{equation}
  By \cref{lem:cs_vee} \cref{item:cs_vee_ml,item:cs_vee_Dirac},
  $k_1^\vee = l_1^\vee$.
  Since $\beta_1(\psi) \psi + \tau(\psi)\vone_\Omega
  \ge \beta_2(\psi)\psi$ for each $\psi \in \Theta$ and
  since
  \begin{equation}
    \gamma_1 - \sum_{\psi \in \Theta} \tau(\psi)
    = \gamma_1 - \sup_{v \in \bU}
      \sum_{\psi \in \Theta} (\beta_2(\psi) - \beta_1(\psi))v(\psi)
    = (\abs{\Theta} + 1)\inf_{v \in \bU} (l_1^\vee(v) - l_2^\vee(v))
      + \gamma_2
    \ge \gamma_2,
  \end{equation}
  applying the monotonicity of $\wpru$ repeatedly gives $k_1 \wpru l_2$.
  Thus, since $m_1 \ipru l_1$, $l_1 \ipru k_1$, and $l_2 \ipru m_2$
  by \cref{lem:aux_imtc}, the transitivity of $\wpru$ implies $m_1 \wpru m_2$.
\end{proof}

The relation $\wpru$ is
\begin{itemize}
  \item
    \emph{regular}
    if it is complete, transitive, and mixture continuous, and
    $s > t$ implies $\Dirac[s \vone_\Omega] \spru \Dirac[t \vone_\Omega]$
    for each $(s, t) \in \SR^2$;
  \item
    \emph{ex ante averse to randomization}
    if $l \wpru m$ implies $l \wpru \lambda l + (1 - \lambda)m$
    for each $(\lambda, (l, m)) \in [0, 1] \times \lott(\Phi)^2$;
  \item
    \emph{independent of constants}
    if $\lambda l + (1 - \lambda)\Dirac[s \vone_\Omega]
    \wpru \lambda m + (1 - \lambda)\Dirac[s \vone_\Omega]$ implies
    $\lambda l + (1 - \lambda)\Dirac[t \vone_\Omega]
    \wpru \lambda m + (1 - \lambda)\Dirac[t \vone_\Omega]$
    for each $(\lambda, (l, m), (s, t))
    \in [0, 1] \times \lott(\Phi)^2 \times \SR^2$.
\end{itemize}

\begin{lemma}\label{lem:aux_rep_u}
  If $\wpru$ is
  regular, monotone, indifferent to mixture timing of constants,
  ex ante averse to randomization, and independent of constants,
  then there exists a normalized convex niveloid
  $W$ on $\Xi$ such that $m \mapsto W(m^\vee)$ represents $\wpru$.
\end{lemma}

\begin{proof}
  Suppose that $\wpru$ is
  regular, monotone, indifferent to mixture timing of constants,
  ex ante averse to randomization, and independent of constants.
  By \cref{lem:aux_mon}, it is $\Xi$-monotone.
  By regularity, indifference to mixture timing of constants, and
  $\Xi$-monotonicity,
  there exists a unique real-valued function $W$ on $\Xi$ such that
  $\Dirac[W(m^\vee)\vone_\Omega] \ipru m$ for each $m \in \lott(\Phi)$.
  Again by $\Xi$-monotonicity, $W$ is monotone.
  By regularity, $W$ is normalized, and
  $m \mapsto W(m^\vee)$ represents $\wpru$.
  For each $(\xi, t) \in \Xi \times \SR$,
  since letting $s = W(\xi + t \vone_\bU) - t$ gives
  $\half(2 s) + \half(2 t) = W(\half(2\xi) + \half(2 t)\vone_\bU)$,
  the independence of constants of $\wpru$ implies
  $s = \half(2s) + \half 0 = W(\half(2\xi) + \half 0) = W(\xi)$.
  Hence, $W$ is translation equivariant,
  so it is a niveloid by \cref{lem:nv_mon_ts} \cref{item:nv_mon_ts_tube}.
  By the ex ante aversion to randomization of $\wpru$,
  the function $W$ is quasiconvex.
  Therefore, applying the same argument as in the proof of Lemma 9
  of \citet{CMMR2014} shows that $W$ is convex.
\end{proof}

For each vNM function $u$ and each $P \in \lott(\cF)$,
let $P_u$ be the pushforward of $P$ under $f \mapsto u \cmpf f$.

\begin{lemma}\label{lem:aux_pf}
  If $u$ is a surjective vNM function,
  then $P \mapsto P_u$ from $\lott(\cF)$ to $\lott(\Phi)$
  is surjective and mixture linear.
\end{lemma}

\begin{proof}
  Let $u$ be a surjective vNM function.
  The surjectivity follows from the surjectivity
  of $f \mapsto u \cmpf f$ from $\cF$ to $\Phi$.
  For each $(\lambda, (P, Q), \Theta)
  \in [0, 1] \times \lott(\cF)^2 \times 2^\Phi$,
  letting $F = \set{f \in \cF \mvert u \cmpf f \in \Theta}$ gives
  $[\lambda P + (1 - \lambda)Q]_u(\Theta)
  = [\lambda P + (1 - \lambda)Q](F)
  = \lambda P(F) + (1 - \lambda)Q(F)
  = \lambda P_u(\Theta) + (1 - \lambda)Q_u(\Theta)$.
  Thus, $P \mapsto P_u$ is mixture linear.
\end{proof}

For each binary relation $\rel'$ on $\lott(\cF)$,
a binary relation is \emph{more indecisive than $\rel'$}
if its graph is included in the graph of $\rel'$.

\begin{lemma}\label{lem:mmeu_cdec}
  Let $(u, \bM)$ and $(u', \bM')$ be
  multi-{\meu} representations of $\rel$ and $\rel'$, respectively.
  Then, $\rel$ is more indecisive than $\rel'$ if and only if
  $u \pat u'$ and $\bM \supseteq \bM'$.
\end{lemma}

\begin{proof}
  The necessity of the more indecisiveness of $\rel$ follows by definition.
  For the sufficiency, suppose that $\rel$ is more indecisive than $\rel'$.
  Since the restrictions of $\rel$ and $\rel'$ to $\Dirac[\lott(X)]$
  are complete,
  the uniqueness of mixture linear representations implies $u \pat u'$.
  Seeking a contradiction, suppose $\bM \not\supseteq \bM'$.
  Let $M' \in \bM' \setminus \bM$.
  By \cref{lem:cs_bU} and \cref{lem:cs_supp} \cref{item:cs_supp},
  the set $\set{H_M|_\Psi \mvert M \in \bM}$
  is a compact convex subset of $\bU$ that does not contain $H_{M'}|_\Psi$.
  Thus, by the separation theorem \citep[Corollary 5.80]{AliprantisBorder2006},
  there exists $(\Lambda, t) \in \ca(\Psi) \times \SR$ such that
  $\inf_{M \in \bM} \int H_M|_\Psi \diff \Lambda
  > t > \int H_{M'}|_\Psi \diff \Lambda$.
  Let $a = \max \{\Lambda^+(\Psi), \Lambda^-(\Psi)\}$, which is positive.
  Define the Borel probability measures $l$ and $m$ on $\Phi$ by
  \begin{align}
    l(\Theta)
    &= \frac{\Lambda^+(\Theta \cap \Psi)}{2a}
      + \Bigl(1 - \frac{\Lambda^+(\Psi)}{2a}\Bigr)\vone_\Theta(0), \\
    m(\Theta)
    &= \frac{\Lambda^-(\Theta \cap \Psi)}{2a}
      + \Bigl(1 - \frac{\Lambda^-(\Psi)}{2a}\Bigr)
      \vone_\Theta\Bigl(\frac{t}{2a - \Lambda^-(\Psi)}\vone_\Omega\Bigr),
  \end{align}
  where $\vone_\Theta$ is the indicator function of $\Theta$ on $\Phi$
  for each Borel subset $\Theta$ of $\Phi$.
  Then, for each $\tilde M \in \bM \cup \{M'\}$,
  by the normalizedness of $H_{\tilde M}$,
  \begin{gather}
    \int H_{\tilde M} \diff l
    = \int H_{\tilde M} \vone_\Psi \diff l
      + \int H_{\tilde M} \vone_{\Psi^\setc} \diff l
    = \frac{1}{2a}\int H_{\tilde M} \diff \Lambda^+, \\
    \begin{aligned}
      \int H_{\tilde M} \diff m
      &= \int H_{\tilde M} \vone_\Psi \diff m
        + \int H_{\tilde M} \vone_{\Psi^\setc} \diff m \\
      &= \frac{1}{2a}\int H_{\tilde M} \diff \Lambda^-
        + \Bigl(1 - \frac{\Lambda^+(\Psi)}{2a}\Bigr)
        m\Bigl(\Bigl\{\frac{t}{2a - \Lambda^-(\Psi)}\vone_\Omega\Bigr\}\Bigr)
      = \frac{1}{2a}\Bigl(\int H_{\tilde M} \diff \Lambda^- + t\Bigr).
    \end{aligned}
  \end{gather}
  Thus, $\int H_M \diff l > \int H_M \diff m$ for each $M \in \bM$, and
  $\int H_{M'} \diff l < \int H_{M'} \diff m$.
  Since $\lott(\Phi)$ is dense
  in the space of Borel probability measures on $\Phi$
  \citep[Theorem 15.10]{AliprantisBorder2006},
  it is without loss of generality to assume that
  $l$ and $m$ are finitely supported.
  By \cref{lem:aux_pf},
  there exists $(P, Q) \in \lott(\cF)^2$ such that $(P_u, Q_u) = (l, m)$.
  Then, $P \rel Q$ and $P \not\rel' Q$,
  which contradicts the more indecisiveness of $\rel$ than $\rel'$.
\end{proof}

\begin{lemma}\label{lem:mmeu_cben}
  Let $\bM$ and $\bM'$ be compact convex subsets of $\bK$.
  Then, $\bM \usd \bM'$ if and only if
  $\max_{M \in \bM} \int H_M \diff m
  \ge \max_{M' \in \bM'} \int H_{M'} \diff m$ for each $m \in \lott(\Phi)$.
\end{lemma}

\begin{proof}
  If $\bM \usd \bM'$,
  then for each $m \in \lott(\Phi)$,
  letting $\bar M' \in \argmax_{M' \in \bM'} \int H_{M'} \diff m$ gives
  $\bar M \subseteq \bar M'$ for some $\bar M \in \bM$, so
  $\max_{M' \in \bM'} \int H_{M'} \diff m
  = \int H_{\bar M'} \diff m
  \le \int H_{\bar M} \diff m
  \le \max_{M \in \bM} \int H_M \diff m$
  by \cref{lem:cs_supp} \cref{item:cs_supp_mon}.
  For the converse, we show the contrapositive.
  Suppose that there exists $M' \in \bM'$ such that
  $M \not\subseteq M'$ for each $M \in \bM$.
  By \cref{lem:cs_bU,lem:cs_supp},
  the set $\set{H_M|_\Psi - H_{M'}|_\Psi \mvert M \in \bM}$
  is a compact convex subset of $\bU$ that is disjoint from $\Cb(\Psi)_+$.
  Thus, by the separation theorem \citep[Theorem 5.79]{AliprantisBorder2006},
  there exists $(\Lambda, t) \in \ca(\Psi) \times \SR$ such that
  $\inf_{h \in \Cb(\Psi)_+} \int h \diff \Lambda
  > t > \max_{M \in \bM} \int H_M|_\Psi \diff \Lambda
  - \int H_{M'}|_\Psi \diff \Lambda$,
  which implies $(\Lambda, t) \in \ca(\Psi)_+ \times (-\SR_{++})$.
  Since $\Lambda \ne 0$, we have $\Lambda(\Psi) > 0$.
  Define the Borel probability measure $m$ on $\Phi$ by
  $m(\Theta) = \Lambda(\Theta \cap \Psi)/\Lambda(\Psi)$.
  Then, $\int H_{M'} \diff m > \max_{M \in \bM} \int H_M \diff m$.
  Since $\lott(\Phi)$ is dense
  in the space of Borel probability measures on $\Phi$
  \citep[Theorem 15.10]{AliprantisBorder2006},
  it is without loss of generality to assume that $m$ is finitely supported.
\end{proof}

\subsection{Proof of \cref{thm:rep}}

For the necessity of the axioms, suppose that $\wpr$ has
a costly ambiguity perception representation $\tuple{u, (\bM, c)}$.
The necessity of \axmref{axm:reg} is routine.
By \cref{lem:aux_vnm_surj}, \axmref{axm:ubd} holds.
Define $W \colon \Xi \to \SR$ by
$W(P_u^\vee) = \max_{M \in \bM} (\int H_M \diff P_u - c(M))$,
which is a convex niveloid by \cref{lem:cs_nv} \cref{item:cs_nv}.
Since $P \mapsto P_u^\vee$ is mixture linear
by \cref{lem:cs_vee} \cref{item:cs_vee_ml} and \cref{lem:aux_pf},
the function $P \mapsto W(P_u^\vee)$ is convex.
Thus, \axmref{axm:eaar} is satisfied.
Since $W([\lambda P + (1 - \lambda)\Dirac[p]]_u^\vee)
= W(\lambda P_u^\vee + (1 - \lambda)u(p)\vone_{\bU})
= W(\lambda P_u^\vee) + (1 - \lambda)u(p)$
for each $(\lambda, P, p) \in [0, 1] \times \lott(\cF) \times \lott(X)$,
\axmref{axm:ica} holds.
The remaining axioms,
\axmref{axm:mon,axm:aepr,axm:irtc},
follow from the monotonicity, concavity, and translation equivariance
of support functions, respectively.

For the sufficiency, assume all the axioms.
By \cref{lem:aux_vnm_surj},
the restriction of $\wpr$ to $\Dirac[\lott(X)]$ is represented
by a surjective vNM function $u$.
Define the relation $\wpru$ on $\lott(\Phi)$ by $P_u \wpru Q_u$ if $P \wpr Q$.
By \axmref{axm:mon}, $\wpru$ is well-defined.
By \cref{lem:aux_pf}, $\wpru$ is
regular, monotone, indifferent to mixture timing of constants,
ex ante averse to randomization, and independent of constants.
Thus, by \cref{lem:aux_rep_u},
there exists a normalized convex niveloid $W$ on $\Xi$
such that $P \mapsto W(P_u^\vee)$ represents $\wpr$.
By \axmref{axm:mon,axm:aepr},
$\rel^*_W$ is $\Phi$-monotone and attracted to ex post randomization.
Hence, by \cref{lem:cs_nv_rep}, there exists a cost structure $(\bM, c)$
such that for each $P \in \lott(\cF)$,
\begin{equation}
  W(P_u^\vee)
  = \max_{M \in \bM} \Bigl(\int H_M \diff P_u - c(M)\Bigr)
  = \max_{M \in \bM} \Bigl[
    \int\Bigl(\min_{\mu \in M} \int u \cmpf f \diff \mu\Bigr) \diff P(f)
    - c(M)\Bigr],
\end{equation}
which shows that $\tuple{u, (\bM, c)}$
is a costly ambiguity perception representation of $\wpr$.
\qed

\subsection{Proof of \cref{prop:mmeu}}

\cref{item:mmeu_rep}\enspace
Let $\tuple{u, (\bM, c)}$ be
a costly ambiguity perception representation of $\wpr$.
Define $W \colon \Xi \to \SR$ by
$W(m^\vee) = \max_{M \in \bM} (\int H_M \diff m - c(M))$.
Then, for each $(P, Q) \in \lott(\cF)^2$, it follows
from \cref{lem:cs_vee} \cref{item:cs_vee_ml} and \cref{lem:aux_pf} that
$P_u^\vee \rel^*_W Q_u^\vee$ if and only if $P \wpr^* Q$.
By \cref{lem:cs_nv} \cref{item:cs_nv}, $W$ is a normalized convex niveloid.
By the monotonicity and concavity of support functions,
$\rel^*_W$ is $\Phi$-monotone and attracted to ex post randomization.
Thus, by \cref{lem:cs_nv_rep}, $\wpr^*$ has a multi-{\meu} representation.

\cref{item:mmeu_uniq}\enspace
Let $(u, \bM)$ and $(v, \bL)$ be multi-{\meu} representations
of a binary relation $\rel$ on $\lott(\cF)$.
Since $\rel$ is more indecisive than itself,
it follows from \cref{lem:mmeu_cdec} that $u \pat v$ and $\bM = \bL$.
\qed

\subsection{Proof of \cref{thm:id}}

Suppose that $u \pat v$, $\bM^* \subseteq \bM$, and
$c = c^\star_{\wpr, u}|_\bM$.
Let $\tuple{\bar u, (\bar \bM, \bar c)}$ be
a costly ambiguity perception representation of $\wpr$.
Since the restrictions of $\wpr$ and $\wpr^*$ to $\Dirac[\lott(X)]$ coincide,
we have $\bar u \pat u \pat v$.
Thus, $\tuple{u, (\bar \bM, \bar c)}$
is also a costly ambiguity perception representation.
Define $W \colon \Xi \to \SR$ by
$W(P_u^\vee) = \max_{M \in \bar \bM} (\int H_M \diff P_u - \bar c(M))$,
which is a convex niveloid by \cref{lem:cs_nv} \cref{item:cs_nv}.
Since $P \mapsto W(P_u^\vee)$ represents $\wpr$,
we have $W(P_u^\vee) = u(\bar P)$ for each $P \in \lott(\cF)$.
Hence, $c^\star_{\wpr, u}(M)
= \sup_{m \in \lott(\Phi)} (\int H_M \diff m - W(m^\vee))$
for each $M \in \bK$.
By \cref{lem:cs_nv_rep}, there exists a compact convex subset $\bar \bM^*$
of $\bK$ such that
\begin{enumerate}
  \item
    for each $(P, Q) \in \lott(\cF)^2$,
    it follows that $P_u^\vee \rel^*_W Q_u^\vee$ if and only if
    $\int H_M \diff P_u \ge \int H_M \diff Q_u$ for each $M \in \bar \bM^*$;
  \item
    for each compact convex subset $\bL$ of $\bK$
    with $\bar \bM^* \subseteq \bL$,
    the pair $(\bL, c^\star_{\wpr, u}|_\bL)$ is a convex cost structure and
    $W(m^\vee) = \max_{L \in \bL} (\int H_L \diff m - c^\star_{\wpr, u}(L))$
    for each $m \in \lott(\Phi)$.
\end{enumerate}
Since for each $(P, Q) \in \lott(\cF)^2$,
it follows from \cref{lem:cs_vee} \cref{item:cs_vee_ml} and \cref{lem:aux_pf}
that $P_u^\vee \rel^*_W Q_u^\vee$ if and only if $P \wpr^* Q$,
we have $\bar \bM^* = \bM^*$ by \cref{prop:mmeu} \cref{item:mmeu_uniq}.
Therefore, $W(P_u^\vee) = \max_{M \in \bM} (\int H_M \diff P_u - c(M))$,
which shows that $\tuple{u, (\bM, c)}$
is a costly ambiguity perception representation of $\wpr$.

For the converse, suppose that $\tuple{u, (\bM, c)}$ is
a convex costly ambiguity perception representation of $\wpr$.
Define $W \colon \Xi \to \SR$ by
$W(m^\vee) = \max_{M \in \bM} (\int H_M \diff m - c(M))$.
Then, $W(P_u^\vee) = u(\bar P)$ for each $P \in \lott(\cF)$.
By \cref{lem:cs_nv} \cref{item:cs_nv_cn}, $c^\diamond = W^*|_{\bM^\diamond}$.
Thus, $c(M) = \sup_{\xi \in \Xi} (\dual{\xi, \pi} - W(\xi))
= \sup_{m \in \lott(\Phi)} (\int H_M \diff m - W(m^\vee))
= c^\star_{\wpr, u}(M)$
for each $(M, \pi) \in \bM \times \bM^\diamond$ with $H_M = \pi^\wedge$.
Define the relation $\wpr'$ on $\lott(\cF)$ by $P \wpr' Q$
if $\int H_M \diff P_u \ge \int H_M \diff Q_u$ for each $M \in \bM$.
Then, for each $(P, Q) \in \lott(\cF)^2$, if $P \wpr' Q$,
then $\lambda P + (1 - \lambda)R \wpr \lambda Q + (1 - \lambda)R$
for each $(\lambda, R) \in (0, 1] \times \lott(\cF)$;
that is, $P \wpr^* Q$.
Hence, by \cref{lem:mmeu_cdec}, $u \pat v$ and $\bM \supseteq \bM^*$.
\qed

\subsection{Proof of \cref{prop:cmp_amb}}

The relation $\wpr_1$ is more ambiguity-averse than $\wpr_2$ if and only if
$u_1 \pat u_2$ and $u_1(\bar P_1) \le u_1(\bar P_2)$.
The inequality is equivalent to
$c^\star_{\wpr_1, u_1} \ge c^\star_{\wpr_2, u_1}$.
\qed

\subsection{Proof of \cref{prop:cmp_inc}}

If $\wpr_1$ is more averse to additional ambiguity than $\wpr_2$,
then $\wpr_1$ is more ambiguity-averse than $\wpr_2$, so
$u_1 \pat u_2$ by \cref{prop:cmp_amb}.
Thus, assume without loss of generality $u_1 = u_2$.
Let $u = u_1$.
For each $i \in \{1, 2\}$, define $W_i \colon \Xi \to \SR$ by
$W_i(m^\vee) = \max_{M \in \bM_i} (\int H_M \diff m - c_i(M))$, and
define $U_i \colon \lott(\cF) \to \SR$ by $U_i(P) = W_i(P_u^\vee)$.

Suppose that $\cC_2(P) \usd \cC_1(P)$ for each $P \in \lott(\cF)$.
By \cref{lem:mmeu_cben}, for each $(P, Q) \in \lott(\cF)^2$,
\begin{equation}\label{eq:mmeu_cben}
  \max_{M \in \cC_2(P)} \int
    \Bigl(\min_{\mu \in M} \int u \cmpf f \diff \mu\Bigr) \diff Q(f)
  \ge \max_{M \in \cC_1(P)} \int
    \Bigl(\min_{\mu \in M} \int u \cmpf f \diff \mu\Bigr) \diff Q(f).
\end{equation}
Choose any $(\lambda, (P, Q), p)
\in [0, 1] \times \lott(\cF)^2 \times \lott(X)$.
Define $R \colon [0, 1] \to \lott(\cF)$ by
$R(t) = \lambda [t P + (1 - t)\Dirac[p]] + (1 - \lambda)Q$.
Let $\tilde P \in \lott(\cF)$ be such that
$\int H_M \diff \tilde P_u = \int H_M \diff P_u - u(p)$ for each $M \in \bM$.
For each $t \in [0, 1]$,
let $(\bar M_1(t), \bar M_2(t)) \in \cC_1(R(t)) \times \cC_2(R(t))$
be such that $\int H_{\bar M_2(t)} \diff \tilde P_u
\ge \int H_{\bar M_1(t)} \diff \tilde P_u$.
Then, since for each $i \in \{1, 2\}$,
the envelope theorem \citep[Theorem 2]{MilgromSegal2002} implies
\begin{equation}
  U_i(\lambda P + (1 - \lambda)Q) - U_i(\lambda \Dirac[p] + (1 - \lambda)Q)
  = U_i(R(1)) - U_i(R(0))
  = \lambda \int_0^1 \Bigl(\int H_{\bar M_i(t)} \diff \tilde P_u\Bigr) \diff t,
\end{equation}
we have
$U_2(\lambda P + (1 - \lambda)Q) - U_2(\lambda \Dirac[p] + (1 - \lambda)Q)
\ge U_1(\lambda P + (1 - \lambda)Q) - U_1(\lambda \Dirac[p] + (1 - \lambda)Q)$.
Thus, $\wpr_1$ is more averse to additional ambiguity than $\wpr_2$.

For the converse,
suppose that $\wpr_1$ is more averse to additional ambiguity than $\wpr_2$.
Choose any $(P, m) \in \lott(\cF) \times \lott(\Phi)$.
Since for each $i \in \{1, 2\}$,
by \cref{lem:cs_nv} \cref{item:cs_nv,item:cs_nv_sd} and
\cref{lem:nv_sd} \cref{item:nv_sd_dd},
\begin{equation}
  \max_{M \in \cC_i(P)} \int H_M \diff m
  = \max_{\pi \in \partial W_i(P_u^\vee) \cap \bM_i^\diamond}
    \dual{m^\vee, \pi}
  = \Diff^+_{m^\vee} W_i(P_u^\vee),
\end{equation}
it suffices to show $\Diff^+_{m^\vee} W_2(P_u^\vee)
\ge \Diff^+_{m^\vee} W_1(P_u^\vee)$ by \cref{lem:mmeu_cben}.
Let $\tilde P \in \lott(\cF)$ be such that $\tilde P_u^\vee = 2 P_u^\vee$.
For each $\lambda \in (0, 1]$,
let $\tilde Q(\lambda) \in \lott(\cF)$ be such that
$\tilde Q(\lambda)_u^\vee = 2 \lambda m^\vee$ and
let $p(\lambda) \in \lott(X)$ be such that
$u(p(\lambda)) = 2(W_1(P_u^\vee + \lambda m^\vee) - W_1(P_u^\vee))$;
then, since
\begin{equation}
  U_1\Bigl(\half \tilde P + \half \tilde Q(\lambda)\Bigr)
    - U_1\Bigl(\half \tilde P + \half \Dirac[p(\lambda)]\Bigr)
  = W_1(P_u^\vee + \lambda m^\vee) - W_1(P_u^\vee) - \half u(p(\lambda))
  = 0,
\end{equation}
we have
\begin{equation}
  W_2(P_u^\vee + \lambda m^\vee) - W_2(P_u^\vee) - \half u(p(\lambda))
  = U_2\Bigl(\half \tilde P + \half \tilde Q(\lambda)\Bigr)
    - U_2\Bigl(\half \tilde P + \half \Dirac[p(\lambda)]\Bigr)
  \ge 0.
\end{equation}
Thus, $\Diff^+_{m^\vee} W_2(P_u^\vee) \ge \Diff^+_{m^\vee} W_1(P_u^\vee)$.
\qed

\subsection{Proof of \cref{prop:cmp_eaar}}

Let $i \in \{1, 2\}$.
Define $W_i \colon \Xi \to \SR$ by $W_i(m^\vee)
= \max_{M \in \bM_i} (\int H_M \diff m - c_i(M))$.
For each $(P, Q) \in \lott(\cF)^2$,
it follows from \cref{lem:cs_nv} \cref{item:cs_nv,item:cs_nv_sd} that
$\cC_i(P) \cap \cC_i(Q) \ne \emptyset$ if and only if
$\partial W_i(P_{u_i}^\vee) \cap \partial W_i(Q_{u_i}^\vee) \cap \bM_i^\diamond
\ne \emptyset$, which is equivalent to
$W_i([\lambda P + (1 - \lambda)Q]_{u_i}^\vee)
= W_i(\lambda P_{u_i}^\vee + (1 - \lambda)Q_{u_i}^\vee)
= \lambda W_i(P_{u_i}^\vee) + (1 - \lambda)W_i(Q_{u_i}^\vee)
= \lambda u_i(\bar P_i) + (1 - \lambda)u_i(\bar Q_i)$
for each $\lambda \in [0, 1]$ by \cref{lem:nv_sd} \cref{item:nv_sd_lin}.
Thus, by \cref{lem:aux_eaar}, the desired equivalence is obtained.
\qed

\subsection{Proof of \cref{cor:rep_meu}}

The necessity of \axmref{axm:eanr} is routine.
For the sufficiency, assume \axmref{axm:eanr}.
We claim that $\wpr$ equals $\wpr^*$.
Let $(P, Q) \in \lott(\cF)^2$.
By definition, $P \wpr^* Q$ implies $P \wpr Q$.
For the converse, suppose $P \wpr Q$.
Choose any $(\lambda, R) \in (0, 1] \times \lott(\cF)$.
By \axmref{axm:reg,axm:mon,axm:irtc},
there exists $(p, q, r) \in \lott(X)^3$ such that
$\Dirac[p] \ipr P$, $\Dirac[q] \ipr Q$, and $\Dirac[r] \ipr R$.
By \axmref{axm:eanr}, applying a similar argument
as in the proof of \cref{lem:aux_eaar} \cref{item:aux_eaar_suff} gives
$\lambda \Dirac[p] + (1 - \lambda)\Dirac[r] \ipr \lambda P + (1 - \lambda)R$ and
$\lambda \Dirac[q] + (1 - \lambda)\Dirac[r] \ipr \lambda Q + (1 - \lambda)R$.
Thus, by \cref{lem:aux_ica} \cref{item:aux_ica_ind},
$\lambda P + (1 - \lambda)R \wpr \lambda Q + (1 - \lambda)R$.

By the claim above and \cref{prop:mmeu} \cref{item:mmeu_rep},
$\wpr$ has a multi-{\meu} representation $(u, \bM)$.
If $\bM$ were not a singleton,
then $\wpr$ would be incomplete by the uniqueness of the {\meu} representation
\citep[Theorem 1 (b)]{GilboaSchmeidler1989}.
Thus, there exists $M \in \bK$ such that $\{M\} = \bM$.
Then, $(u, M)$ is an expected {\meu} representation of $\wpr$.
\qed

\subsection{Proof of \cref{prop:chara_int}}

We have $\mu \in \core(\wpr)$ if and only if
$\int (\int \phi \diff \mu) \diff P_u(\phi)
\ge \int H_M \diff P_u - c(M)$ for each $(M, P) \in \bM \times \lott(\cF)$,
which is equivalent to $\int H_{\{\mu\}} \diff P_u \ge \int H_M \diff P_u$
for each $(M, P) \in \bM \times \lott(\cF)$;
that is, $\mu \in \bigcap \bM$.
\qed

\subsection{Proof of \cref{prop:cmp_meu}}

Define $U \colon \lott(\cF) \to \SR$ by
$U(P) = \max_{M \in \bM} (\int H_M \diff P_u - c(M))$.

\cref{item:cmp_meu_meet} $\implies$ \cref{item:cmp_meu_maaa}.
Assume \cref{item:cmp_meu_meet}.
Denote by $\sbp$ the asymmetric part of $\wbp$.
Choose any $(\lambda, (P, Q), p)
\in [0, 1] \times \lott(\cF)^2 \times \lott(X)$ with
$\lambda \Dirac[p] + (1 - \lambda)Q \sbp \lambda P + (1 - \lambda)Q$.
By \cref{cor:rep_meu}, $\Dirac[p] \sbp P$, so
$v(p) > \int H_L \diff P_v
\ge \max_{M \in \bM} \int H_M \diff P_v$ by \cref{item:cmp_meu_meet}.
Thus, $u(p) > \max_{M \in \bM} \int H_M \diff P_u$, which implies
\begin{align}
  U(\lambda \Dirac[p] + (1 - \lambda)Q)
  &= \lambda u(p)
    + \max_{M \in \bM} \Bigl[(1 - \lambda)\int H_M \diff Q_u - c(M)\Bigr] \\
  &> \max_{M \in \bM} \lambda \int H_M \diff P_u
    + \max_{M \in \bM} \Bigl[
    (1 - \lambda)\int H_M \diff Q_u - c(M)\Bigr] \\
  &\ge \max_{M \in \bM} \Bigl[
    \lambda \int H_M \diff P_u + (1 - \lambda)\int H_M \diff Q_u - c(M)
    \Bigr] \\
  &= U(\lambda P + (1 - \lambda)Q).
\end{align}

\cref{item:cmp_meu_maaa} $\implies$ \cref{item:cmp_meu_aa}.
By definition.

\cref{item:cmp_meu_aa} $\implies$ \cref{item:cmp_meu_meet}.
Assume \cref{item:cmp_meu_aa}.
By \cref{prop:cmp_amb}, $u \pat v$.
By \cref{prop:chara_int},
it suffices to show $L \subseteq \core(\wpr)$.
Choose any $(\ell, P) \in L \times \lott(\cF)$.
Let $p \in \lott(X)$ be such that $\Dirac[p] \ipr P$.
Then, by \cref{item:cmp_meu_aa}, $P \wbp \Dirac[p]$,
which is equivalent to $\int H_L \diff P_v \ge v(p)$.
Since $v(P^\ell) = \int H_{\{\ell\}} \diff P_v \ge \int H_L \diff P_v$ and
since $u \pat v$ and $\Dirac[p] \ipr P$,
we have $u(P^\ell) \ge u(p) = U(P)$.
Hence, $\Dirac[P^\ell] \wpr P$.
\qed

\subsection{Proof of \cref{cor:rep_spec}}

In each statement, the necessity of the axiom is routine.
We show only the sufficiency.
Let $\tuple{u, (\bM, c)}$ be
a costly ambiguity perception representation of $\wpr$.
By \cref{prop:mmeu} and \cref{thm:id},
we may assume that $(u, \bM)$ is a multi-{\meu} representation of $\wpr^*$.
Define $U \colon \lott(\cF) \to \SR$ by
$U(P) = \max_{M \in \bM} (\int H_M \diff P_u - c(M))$.
Let $\ipr^*$ be the symmetric part of $\wpr^*$.
Observe that for each $(\lambda, (f, g)) \in [0, 1] \times \cF^2$,
if $\Dirac[\lambda f + (1 - \lambda)g]
\ipr^* \lambda \Dirac[f] + (1 - \lambda)\Dirac[g]$,
then $\lambda H_M(u \cmpf f) + (1 - \lambda)H_M(u \cmpf g)
= H_M(\lambda(u \cmpf f) + (1 - \lambda)(u \cmpf g))$ for each $M \in \bM$.

\cref{item:rep_oap}\enspace
Assume \axmref{axm:irtcm}.
By the above observation, for each $(\phi, \psi) \in \Phi^2$
if $(\phi(\omega) - \phi(\omega'))(\psi(\omega) - \psi(\omega')) \ge 0$
for each $(\omega, \omega') \in \Omega^2$,
then $H_M(\phi) + H_M(\psi) = H_M(\phi + \psi)$ for each $M \in \bM$.
Thus, by the Theorem and Proposition 3 of \citet{Schmeidler1986},
every member of $\bM$ is the core of a convex capacity.

\cref{item:rep_mh}\enspace
Assume \axmref{axm:irt}.
For each $M \in \bM$,
the function $H_M$ is mixture linear by the above observation,
so it is linear.
Thus, by Theorem 5.54 of \citet{AliprantisBorder2006},
every member of $\bM$ is a singleton.

\cref{item:rep_dseu}\enspace
Assume \axmref{axm:sica}.
Define $\cC \colon \lott(\cF) \coto \bM$ by
$\cC(P) = \argmax_{M \in \bM} (\int H_M \diff P_u - c(M))$.
We show that $\cC(P) \cap c^{-1}(\{0\}) \ne \emptyset$
for each $P \in \lott(\cF)$,
in which case $(u, \bM \cap c^{-1}(\{0\}))$ is
a dual-self expected utility representation of $\wpr$.
Choose any $P \in \lott(\cF)$.
Let $p \in \lott(X)$ be such that $\Dirac[p] \ipr P$,
let $Q = \half P + \half \Dirac[p]$, and
let $M \in \cC(Q)$.
Then, $c(M) \ge \int H_M \diff P_u - U(P)$.
Since $P \ipr Q$ by \axmref{axm:sica}, we have
\begin{equation}
  U(P)
  = U(Q)
  = \half \int H_M \diff P_u + \half u(p) - c(M)
  = \half \int H_M \diff P_u + \half U(P) - c(M),
\end{equation}
so $c(M) = (\int H_M \diff P_u - U(P))/2$.
Thus, $\int H_M \diff P_u - U(P) = 0$,
which implies $M \in \cC(P) \cap c^{-1}(\{0\})$.
\qed

\subsection{Proof of \cref{prop:kz_axm}}

Assume \axmref{axm:reg,axm:aepr,axm:irtc,axm:eapr,axm:sica}.
Choose any $(\lambda, (f, g)) \in (0, 1) \times \cF^2$.
If $\Dirac[f] \wpr \Dirac[g]$,
then $\Dirac[\lambda f + (1 - \lambda)g]
\wpr \lambda \Dirac[f] + (1 - \lambda)\Dirac[g]
\wpr \Dirac[g]$ by \axmref{axm:aepr,axm:eapr}.
Also, by \axmref{axm:irtc,axm:sica},
$\Dirac[f] \wpr \Dirac[g]$ if and only if
$\Dirac[\lambda f + (1 - \lambda)p]
\ipr \lambda \Dirac[f] + (1 - \lambda)\Dirac[p]
\wpr \lambda \Dirac[g] + (1 - \lambda)\Dirac[p]
\ipr \Dirac[\lambda g + (1 - \lambda)p]$ for each $p \in \lott(X)$.
\qed

\section{Machina's examples}\label{sec:machina}

\subsection{50--51 example}\label{subsec:5051}

Consider \citets{Machina2009} thought experiment
with a box containing $101$ balls, called the $50$--$51$ example.
Out of the balls, $50$ are either red or blue, and
$51$ are either green or purple.
A ball is drawn at random from the box.
The {\dm} is offered four acts $f_5$, \dots, $f_8$.
Each act pays off according to the color of the drawn ball,
as described in \cref{tab:5051}.

\begin{table}[h]
  \centering
  \renewcommand{\arraystretch}{1.3}
  \begin{tabularx}{80mm}{CCCCC}\toprule
          & \multicolumn{2}{c}{$50$ balls} & \multicolumn{2}{c}{$51$ balls} \\
          \cmidrule(lr){2-3}\cmidrule(lr){4-5}
          & Red   & Blue  & Green  & Purple \\ \hline
    $f_5$ & $200$ & $200$ & $100$  & $100$  \\
    $f_6$ & $200$ & $100$ & $200$  & $100$  \\
    $f_7$ & $300$ & $200$ & $100$  & $0$    \\
    $f_8$ & $300$ & $100$ & $200$  & $0$    \\ \bottomrule
  \end{tabularx}
  \caption{$50$--$51$ example.}
  \label[table]{tab:5051}
\end{table}

The only difference between $f_5$ and $f_6$, as well as $f_7$ and $f_8$, is
which draw of a blue or green ball leads to a higher payoff of $200$.
Due to the $51^\text{st}$ ball,
when the likelihood of blue and green is assessed,
$f_6$ and $f_8$ have a slight ``objective advantage''
compared to $f_5$ and $f_7$, respectively.
\citet{Machina2009} conjectures that plausible choices are
$f_5 \spr f_6$ and $f_8 \spr f_7$.
The first preference is motivated by the fact that
$f_6$ is more ambiguous than $f_5$---%
ambiguity aversion offsets the objective advantage of $f_6$.
The second might arise
because both $f_7$ and $f_8$ are ambiguous---%
$f_7$ does not have an informational advantage as $f_5$ does.
However, in most ambiguity-sensitive models
(the {\meu} model, the $\alpha$-{\meu} model \citep{GMM2004},
the Choquet expected utility model \citep{Schmeidler1989},
the variational model \citep{MMR2006Ecta},
the uncertainty averse model \citep{CMMM2011}, and
the smooth ambiguity model \citep{KMM2005}),
$f_5 \spr f_6$ implies $f_7 \spr f_8$ \citep{Machina2009,BLP2011}.
In \cref{subsec:dseu}, we further extend the incompatibility
to preferences satisfying \textit{certainty independence}.

In a coeval work by \citet{PTX2026},
they provide an example of an optimal ambiguity attitude preference
that rationalizes \citets{Machina2009} conjecture.
Since their model is a special case of ours,
the costly ambiguity perception model can also rationalize the conjecture.
Here, to obtain a clearer intuition,
we further examine the relationship between our model and the choices
using a more general parametrization, as in \cref{ex:ref}.
Let $p = 50/101$, let $q = 1 - p$,
let $B = [-p/2, p/2]$, and let $G = [-q/2, q/2]$.
Identify each $(\mu_\rmB, \mu_\rmG) \in [0, p] \times [0, q]$
with the prior such that
the probability of drawing blue is $\mu_\rmB$ and drawing green is $\mu_\rmG$.
For each $(\beta, \gamma) \in [0, 1]^2$,
let $M(\beta, \gamma) = \set{(p/2 + \beta b, q/2 + \gamma g)
\mvert (b, g) \in B \times G}$.
Let $\bM = \set{M(\beta, \gamma) \mvert (\beta, \gamma) \in [0, 1]^2}$,
let $c \colon [0, 1]^2 \to \SR_+$ be
a strictly decreasing lower semicontinuous grounded function, and
let $\theta \in (0, \infty)$.
Define $\phi_\theta \colon \SR \to \SR$ by
$\phi_\theta(t) = \min_{(\beta, \gamma) \in [0, 1]^2}
[t(p\beta + q\gamma) + \theta c(\beta, \gamma)]$.
Let $U$ be the utility function over acts
corresponding to the cost structure $(\bM, \theta c)$.
Then,
\begin{align}
  U(f_5)
  &= 200p + 100q = 100(1 + p), \\
  U(f_6)
  &= \max_{(\beta, \gamma) \in [0, 1]^2} \Bigl[100
    \min_{(b, g) \in B \times G} \Bigl(\frac{3}{2} - \beta b + \gamma g\Bigr)
    - \theta c(\beta, \gamma)\Bigr]
  = 150 - \phi_\theta(50), \\
  U(f_7)
  &= \max_{(\beta, \gamma) \in [0, 1]^2} \Bigl[100
    \min_{(b, g) \in B \times G}
    \Bigl(\frac{5}{2}p + \half q - \beta b + \gamma g\Bigr)
    - \theta c(\beta, \gamma)\Bigr]
  = 50 + 200p - \phi_\theta(50), \\
  U(f_8)
  &= \max_{(\beta, \gamma) \in [0, 1]^2} \Bigl[100
    \min_{(b, g) \in B \times G} (2p + q - 2\beta b + 2\gamma g)
    - \theta c(\beta, \gamma)\Bigr]
  = 100 + 100p - \phi_\theta(100).
\end{align}
Thus, $f_5 \wpr f_6$ if and only if $\phi_\theta(50) \ge 50 - 100p$, and
$f_7 \wpr f_8$ if and only if
$\phi_\theta(100) - \phi_\theta(50) \ge 50 - 100p$.
For each $t \in \{50, 100\}$,
the function $\theta \mapsto \phi_\theta(t)$ is concave and nondecreasing, and
$\lim_{\theta \to \infty} \phi_\theta(t) = t$ and
$\lim_{\theta \downarrow 0} \phi_\theta(t) = 0$.
Hence, there exists $\theta^* \in \SR_+$ such that
$\theta < \theta^*$ implies
$\phi_\theta(50) < 50 = 50(p + q) + \theta c(1, 1)$, in which case
every $(\bar \beta, \bar \gamma) \in \argmin_{(\beta, \gamma) \in [0, 1]^2}
[50(p\beta + q\gamma) + \theta c(\beta, \gamma)]$ satisfies
\begin{equation}
  2\phi_\theta(50)
  = 100(p\bar \beta + q\bar \gamma) + 2 \theta c(\bar \beta, \bar \gamma)
  > 100(p\bar \beta + q\bar \gamma) + \theta c(\bar \beta, \bar \gamma)
  \ge \phi_\theta(100),
\end{equation}
so $\phi_\theta(50) > \phi_\theta(100) - \phi_\theta(50)$.
Therefore, there exists $(\ubar \theta, \bar \theta) \in \SR_{++}^2$ such that
$\ubar \theta < \bar \theta < \theta^*$ and
$\theta \in (\ubar \theta, \bar \theta)$ implies
$\phi_\theta(50) > 50 - 100p > \phi_\theta(100) - \phi_\theta(50)$,
in which case $f_5 \spr f_6$ and $f_8 \spr f_7$.

If $\theta$ is not too small ($\theta > \ubar \theta$),
then the {\dm} does not choose very small $(\beta, \gamma)$ at $f_6$,
so its objective advantage over $f_5$ due to the $51^\text{st}$ ball
does not offset ambiguity aversion.
If $\theta$ is not too high ($\theta < \bar \theta$),
then the {\dm} is not extremely pessimistic
at each ``cell'' of unambiguous events,
so the objective advantage of $f_8$ dominates
the pessimism that makes $f_7$ favorable.

\subsection{Dual-self expected utility model}\label{subsec:dseu}

We explore the implications of the dual-self expected utility model
on the behavior in Machina's examples.
As in \cref{ex:ref} and the example in \cref{subsec:5051},
for each example, identify each $(\mu_\rmB, \mu_\rmG)$
with the prior over colors such that
the probability of drawing blue is $\mu_\rmB$ and drawing green is $\mu_\rmG$.

The dual-self expected utility model can explain the pattern
$f_2 \spr f_1$ and $f_3 \spr f_4$ in the reflection example (\cref{tab:ref}).
For example, let $M_\rmB = \{1/4\} \times [0, 1/2]$,
let $M_\rmG = [0, 1/2] \times \{1/4\}$, and
let $\bM = \{M_\rmB, M_\rmG\}$.
Let $U$ be the utility function over acts corresponding to the set $\bM$
of feasible ambiguity perceptions.
Then,
\begin{alignat}{2}
  U(f_1)
  &= \max_{k \in \{\rmB, \rmG\}}
    \min_{(b, g) \in M_k} 100\Bigl(\half + b + g\Bigr)
  = 75, &\qquad
  U(f_2)
  &= \max_{k \in \{\rmB, \rmG\}}
    \min_{(b, g) \in M_k} 100\Bigl(\half + 2g\Bigr)
  = 100, \\
  U(f_3)
  &= \max_{k \in \{\rmB, \rmG\}}
    \min_{(b, g) \in M_k} 100\Bigl(2b + \half\Bigr)
  = 100, &\qquad
  U(f_4)
  &= \max_{k \in \{\rmB, \rmG\}}
    \min_{(b, g) \in M_k} 100\Bigl(b + g + \half\Bigr)
  = 75,
\end{alignat}
so $f_2 \spr f_1$ and $f_3 \spr f_4$.
In this example, each of $M_\rmB$ and $M_\rmG$ is an ambiguity perception
at which the {\dm} perceives ambiguity
only about one of red--blue or green--purple.
Since the {\dm} cannot reduce ambiguity perception
about both of those two dimensions at the same time,
the ``more ambiguous'' acts $f_1$ and $f_4$ are evaluated lower.
This parameterization is compatible with the typical Ellsberg-pattern
as it is absolutely ambiguity-averse
from \cref{cor:absaa} or Proposition 2 of \citet{CFIL2022}.

In contrast, it fails to accommodate $f_5 \spr f_6$ and $f_8 \spr f_7$
in the $50$--$51$ example (\cref{tab:5051}).
Suppose that $\wpr$ has a dual-self expected utility representation $(u, \bM)$.
Define auxiliary acts $g$, $h$, and $p$ as described in \cref{tab:aux}.
Then,
\begin{equation}
  f_5 = \frac{1}{3}g + \frac{2}{3}p, \qquad
  f_6 = \frac{1}{3}h + \frac{2}{3}p, \qquad
  f_7 = \frac{2}{3}g + \frac{1}{3}h, \qquad
  f_8 = \frac{1}{3}g + \frac{2}{3}h.
\end{equation}
Suppose $f_5 \spr f_6$.
Since $\wpr$ satisfies the \textit{certainty independence} axiom
\citep[Theorem 1]{CFIL2022}, we have $g \spr h$.
Thus, since for each $\lambda \in [0, 1]$,
the corresponding utility function $U$ over acts satisfies
\begin{align}
  U(\lambda g + (1 - \lambda)h)
  &= \max_{M \in \bM} \min_{\mu \in M} \Bigl[
    \lambda \int u \cmpf g \diff \mu
    + (1 - \lambda)\int u \cmpf h \diff \mu\Bigr] \\
  &= \max_{M \in \bM} \min_{\mu \in M} \Bigl[
    \lambda \Bigl(300 \times \frac{50}{101}\Bigr)
    + (1 - \lambda)\int u \cmpf h \diff \mu\Bigr] \\
  &= \lambda \Bigl(300 \times \frac{50}{101}\Bigr)
    + (1 - \lambda)\max_{M \in \bM} \min_{\mu \in M} \int u \cmpf h \diff \mu
  = \lambda U(g) + (1 - \lambda)U(h),
\end{align}
we have $f_7 \spr f_8$.

\begin{table}[h]
  \centering
  \renewcommand{\arraystretch}{1.3}
  \begin{tabularx}{80mm}{CCCCC}\toprule
        & \multicolumn{2}{c}{$50$ balls} & \multicolumn{2}{c}{$51$ balls} \\
        \cmidrule(lr){2-3}\cmidrule(lr){4-5}
        & Red   & Blue  & Green  & Purple \\ \hline
    $g$ & $300$ & $300$ & $0$    & $0$    \\
    $h$ & $300$ & $0$   & $300$  & $0$    \\
    $p$ & $150$ & $150$ & $150$  & $150$  \\ \bottomrule
  \end{tabularx}
  \caption{Auxiliary acts.}
  \label[table]{tab:aux}
\end{table}

\bibliographystyle{refs}
\bibliography{refs}

\end{document}